\documentclass[journal]{IEEEtran}
\usepackage[cmex10]{amsmath}
\usepackage{dsfont}
\usepackage{graphicx}
\usepackage{adjustbox}
\usepackage{url, multicol, multirow, color, enumerate, verbatim, amsfonts,subeqnarray}
\usepackage{threeparttable}
\usepackage[ruled,linesnumbered,vlined]{algorithm2e}
\usepackage{rotating}
\usepackage{amsmath,bm}
\usepackage{tikz}
\usetikzlibrary{fit}
\usepackage{tabularx}
\usepackage{stfloats}
\usepackage{booktabs}
\usepackage{romannum}
\usepackage{threeparttable}
\usepackage{multirow}
\usepackage{amsthm}
\usepackage{pifont}
\usepackage{cite}
\usetikzlibrary{patterns}
\usetikzlibrary{shapes.geometric, arrows}
\usetikzlibrary{shapes,arrows,chains}
\usepackage{amsfonts, verbatim, amsmath, color, tikz, subeqnarray}
\usetikzlibrary{arrows.meta}
\usepackage{enumerate}
\usepackage{amssymb}
\SetKwRepeat{Do}{do}{while}%

\usetikzlibrary{shapes,arrows,snakes}
\newcommand{\tabincell}[2]{\begin{tabular}{@{}#1@{}}#2\end{tabular}}  
\newcommand{\dualvector}{\textcolor{black}}

\newcommand{\rsred}{\textcolor{black}}
\newcommand{\yu}{\textcolor{black}}

\newcommand{\rrred}{\textcolor{black}}

\newcommand{\revise}{\textcolor{black}}
\newcommand{\serevise}{\textcolor{black}}
\newcommand{\threvise}{\textcolor{black}}

\newcommand{\gred}{\textcolor{black}}

\newcommand{\sred}{\textcolor{black}}

\newcommand{\ggy}{\textcolor{black}}

\newcommand{\pb}{\overline{P}}

\newcommand{\pl}{\underline{P}}

\newcommand{\V}{V}
\newcommand{\Vb}{\overline{V}}

\newcommand\Mm{{\mathcal {M}}}

\def\Xe{{\mathcal{X}}}
\def\R{{\mathbb{R}}}
\def\De{{\mathcal D}}

\newcommand {\Rmnum} [1] {\expandafter\@slowromancap\romannumeral#1\@}

\newtheorem{remark}{Remark}

\newcommand{\ignore}[1]{}

\newcommand{\rere}{\textcolor{black}}
\newtheorem{prop}{Proposition}
\newtheorem{lemma}{Lemma}

\def\Z{{\mathbb{Z}}}

\def\K {{\cal K}}

\def\B{{\mathbb{B}}}

\def\R{{\mathbb{R}}}

\def\Z{{\mathbb{Z}}}

\def\T{{\mathcal T}}

\def\De{{\mathcal D}}

\def\G{{\mathcal G}}

\newtheorem{theorem}{Theorem}
\def\Xe{{\mathcal{X}}}

\makeatletter
\def\endthebibliography{%
	\def\@noitemerr{\@latex@warning{Empty `thebibliography' environment}}%
	\endlist
}
\makeatother
\begin{document} \title{{An Extended Integral Unit Commitment Formulation and an Iterative Algorithm for Convex Hull Pricing}}
	
	%\markboth{\threvise{Third revision} submitted to \emph{IEEE Transactions on Power Systems}, \today}
	%{Shell \MakeLowercase{\textit{et al.}}: Bare Demo of IEEEtran.cls for Journals}
	
	\author{Yanan~Yu, %~\IEEEmembership{Student Member,~IEEE;} 
		Yongpei~Guan, %~\IEEEmembership{Senior Member,~IEEE;}
		Yonghong~Chen %,%~\IEEEmembership{Senior Member,~IEEE}
		
		%        Yonghong~Chen %~\IEEEmembership{Member,~IEEE;}
		%  Xing~Wang,~\IEEEmembership{Senior Member,~IEEE}
		
\thanks{Y. Yu and Y. Guan are with the University of Florida, Gainesville, FL 32611. Y. Chen \sred{(Consultant Advisor)} is with MISO, Carmel, IN 46032. \revise{A preliminary study of this work is shown in an earlier version of this paper~\cite{yyu2019ieee}.}}
	}
	
	\maketitle
	\vspace{-2.0cm}
	
\begin{abstract}
To increase market transparency, independent system operators (ISOs) have been working on minimizing uplift payments based on convex hull pricing theorems. \rere{However,} the large-scale complex systems for ISOs bring computational challenges to the existing convex hull pricing algorithms. In this paper, based on the analysis of \rere{specific} generator features in \rere{the} Midcontinent ISO (MISO) system, besides reviewing integral formulations for several special cases, we develop two integral formulations of a single generator that can capture these features. We then build a compact convex hull pricing formulation based on these integral formulations. \serevise{Meanwhile,} to improve the computational efficiency, we propose innovative iterative algorithms with convergence properties, plus a complementary algorithm, to obtain a convex hull price. \rere{The computational results indicate that \serevise{our approach leads to an exact convex hull price} on MISO instances \ggy{with and without} transmission constraints and the solutions can be obtained within $20$ minutes}. 
\end{abstract}
	
	\begin{IEEEkeywords}
		Convex Hull Pricing, Iterative Algorithm, Integral Formulation.
	\end{IEEEkeywords}

	\section*{Nomenclature}
	%\subsection{Sets and Dimensions}
	\noindent\textit{A. Set and Dimension}
	\begin{description}[\IEEEsetlabelwidth{$(\cdot)_r + \textbf{i}_i$}]
		\item[$\mathcal{G}$] Set of all generators.
		%\item[$\mathcal{M}_k$] Set of modes in the state transition graph for combined-cycle unit $k$.
		% \item[$\mathcal{M}_m^{\text{in}}$] Set of modes that can transform to mode $m$ in one step.
		% \item[$\mathcal{M}_{m}^{\text{out}}$] Set of modes that can get from mode $m$ in one step.
		% \item[$\text{CTS,STS}$] Set of CTs and STs of a CCU unit, respectively.
		\item[$\mathcal{T}$] Set of operation time span.
		\item[$\mathcal{S}$] Set of start-up states, i.e., $\mathcal{S}$ = \{h(hot), w(warm), c(cold)\}. 
		\item[$\mathcal{B}$] Set of buses.
	\end{description}
	\textit{B. Parameters}
	\begin{description}[\IEEEsetlabelwidth{$(\cdot)_r + \textbf{i}_i$}]
		%\item[$T$] Operation time span. (h)
		\item[$D_t$] Load at time $t$ (MW). 
		\item[${\ggy{F_{is}}}$] Start-up cost of generator $i$ in state $s \in \mathcal{S}$ (\$). 
		\item[${\ggy{G_{it}}}$] No load cost of generator $i$ at time $t$ (\$). 
		\item[${\ggy{Q_{it}^{k}}}$] Piecewise linear production cost interception coefficient of generator $i$ \rsred{in} segment $k$ at time $t$ (\$). 
		\item[${\ggy{H_{it}^{k}}}$] Piecewise linear production cost slope coefficient of generator $i$ \rsred{in} segment $k$ at time $t$ (\$/MWh). 
		\item[$\pl_{it}/\pb_{it}$] Min/max generation amount of generator $i$ at time $t$ (MW). 
		\item[$L_{i}/\ell_{i}$] Minimum-up/-down time limit of generator $i$ (h). 
		%\item[$\ell_i$] minimum downtime of generator $i$. (h)
		\item[$\overline{L}_i$] Maximum-up time of generator $i$ (h). 
		\item[$\V_{it}^u/\V_{it}^e$] {Ramp-up/-down rate} of generator $i$ at time $t$ (MW/h). 
		\item[$\Vb_{it}^u/\Vb_{it}^e$] Start-up/shut-down ramp rate of generator $i$ at time $t$ (MW/h). 
		\item[$T_i^{w}/T_i^{c}$] Down-time \rsred{limit} for warm/cold start. \ggy{The start-up cost is hot-start cost if the shut-down time is less than $T_i^w$, \rsred{warm-start} cost if it is longer than $T_i^w$ and less than $T_i^{c}$, and \rsred{cold-start} cost if it is longer than $T_i^c$ (h). }
		\item[\revise{$\underline{T}_i^s$}] \revise{Minimum number of OFF periods for generator $i$ in start-up state $s \in \mathcal{S}$ (h). }
		\item[\revise{$\overline{T}_i^s$}] \revise{Maximum number of OFF periods for generator $i$ in start-up state $s \in \mathcal{S}$ (h). }		
		%\item[$\kapp{\ggy{G_{it}}}/\varpi_{it} $] The enforced min-up/-down time parameter of generator $i$ at time $t$. If $\kappa_{it} = 0$, the min-up time limit of generator $i$ is not forced at time $t$. 
		\item[$S'_i\hspace{-0.03in}\rsred{(t)}\hspace{-0.03in}/\hspace{-0.03in}S_i\hspace{-0.03in}(t)$] Time-dependent start-up/shut-down \rsred{cost} for generator $i$ (\$). 
		\item[$S'_i/S_i$] Constant start-up/shut-down cost for generator $i$ (\$). 
		\item[\revise{$\kappa_{t}/\varpi_{t}$}] \revise{Indicator parameter showing whether min-up/-down time constraint is enforced for a generator at time $t$ (``$1$'' if it is yes and ``$0$'' otherwise).}
		\item[\revise{$a_j/b_j$}] \revise{Slope/intercept of segment $j$ in a piecewise linear function {(\$/MWh / \$). }}% the jth piece of the generation cost function}]
		\item[{$N$}] {Number of pieces in a piecewise linear cost function.}
	\end{description}
	\textit{C. Decision Variables}
	\begin{description}[\IEEEsetlabelwidth{$(\cdot)_r + \textbf{i}_i$}]
		%\item[$\textbf{i}$] Imaginary unit.
		\item[$u_{it}$] On/Off status of \yu{generator} $i$ at time $t$.
		\item[$v_{it}$] Start-up status of \yu{generator} $i$ at time $t$.
		\item[$e_{it}$] Shut-down status of \yu{generator} $i$ at time $t$.
		\item[${\ggy{\delta_{it}^s}}$] Indicator variable, which is \ggy{$1$} if generator $i$ starts at time $t$ in state $s \in \mathcal{S}$.
		\item[$x_{it}$] Generation amount of generator $i$ at time $t$ (MW). 
		\item[$f_{it}$] Production cost of generator $i$ at time $t$ (\$). 
		\item[\revise{$\phi_s$}] \revise{Production cost of a generator at time $s$ (\$).} 
		\item[$w_{t}$] Indicator variable, which is $1$ if an initial-on generator shuts down for the first time at time $t+1$.
		%\item[\revise{$w_{T}$}] \revise{Indicator variable, which is $1$ if an initial-on generator keeps online until the end.}
		\item[$y_{tk}$] {Indicator variable, which is $1$ if a generator starts up at time $t$ and shuts down at time $k+1$.}
		%\item[\revise{$y_{tT}$}] \revise{Indicator variable, which is $1$ if a generator starts up at time $t$ and keeps online until the end.}
		\item[$z_{tk}$] Indicator variable, which is $1$ if a generator shuts down at time $t+1$ and starts up at time $k$.
		\item[$\theta_{t}$] Indicator variable, which is $1$ if a generator shuts down at time $t+1$ and never starts up again.
		\item[$q_{tk}^s$] Generation amount of a generator at time $s$ \rsred{when it} starts up at time $t$ and shuts down at time $k+1$ (MW).
		\item[$\phi_{tk}^s$] Production cost of a generator at time $s$ \rsred{when} it starts up at time $t$ and shuts down at time $k+1$ (\$). 
		\item[\serevise{$\bar\gamma_{t(b)}$}] \serevise{Convex hull price at time $t$ (in bus $b$) (\$/MWh). } 
		\item[\serevise{$\beta_{t(b)}$}] \serevise{LMP at time $t$ (in bus $b$) (\$/MWh). } 
	\end{description}
	\IEEEpeerreviewmaketitle
	
\section{Introduction}\label{sec:intro}
In the current \rsred{U.S.} day-ahead electricity market operated by ISOs, unit commitment and economic dispatch (UCED) problems are solved to determine the generation amount of each generator and the \yu{market-clearing} price. 
\dualvector{The locational marginal prices (LMP) for market clearance is calculated based on the optimal dual vector corresponding to the system-wide constraints in the ED problem, in which commitment decision variables are fixed at their optimal values.}
%The optimal dual values \rsred{corresponding to} the load balance constraints in the ED problem, \rsred{in which the commitment decision variables are fixed at their optimal values,} are \rsred{referred} as the local marginal prices (LMPs) \rsred{for market clearance}. 
Since the commitment \rsred{statuses are} fixed, the LMPs cannot cover the start-up and no-load costs of the generators~\cite{wang2016commitment}. {Accordingly}, uplift payments are \rsred{introduced} to make up the possible loss and motivate the \rsred{generation side market participants} to comply with the commitment schedule provided by ISOs. Since the uplift payment cost is \ggy{not transparent}, ISOs \ggy{aim to} minimize the uplift payment cost \rsred{for their daily operations}\revise{~\cite{PJM2014,uplift2014}}. 

\rsred{To minimize the uplift payments, a convex hull pricing approach has been recently introduced and received significant attention~\cite{gribik2007market,schiro2016convex,hua2017convex,conv2019}. This pricing approach aims to minimize the uplift payments over all possible uniform prices. This approach has \yu{the} advantage to provide an optimal uniform price for the system.} However, on the other hand, it requires to obtain an optimal Lagrangian dual multiplier for a mixed-integer UCED problem and has been computationally intractable for market implementation~\cite{wang2016commitment}. 

To overcome this challenge, researchers and practitioners have explored two different approaches to derive prices that could provide good approximations for the exact convex hull price. One approach is to develop efficient algorithms to solve the Lagrangian dual problem of the UCED problem. For instance, in~\cite{wang2013extreme,wang20131}, the authors developed an extreme-point sub-differential method to strengthen the traditional Lagrangian relaxation approach to obtain the price.  In\revise{~\cite{wang2013subgradient}}, a subgradient simplex cutting plane method is developed to solve the dual problem by iteratively cutting off non-optimal solutions. Those algorithms show some benefits, while they can not guarantee convergence in polynomial time. Meanwhile, researchers turn to the other approach, called the primal approach. \revise{This approach utilizes the property that\yu{,} with the convex hull description of each generator and the convex envelope of the cost function, we can solve a relaxed UCED problem, which is a linear program, and derive the convex hull price based on the optimal dual vector associated with the system-wide constraints~\cite{schiro2016convex,hua2017convex}. The work in~\cite{hua2017convex} gets an approximation of the exact convex hull price by deriving a tight linear program relaxation of the convex hull description for each generator. Readers are also referred to  \cite{G2013,frangioni2015,Jim2018} for other tight formulations.} 

\revise{In practice, MISO has implemented an approximation of convex hull pricing based on Lagrangian relaxation, named extended locational marginal prices (ELMPs)~\cite{wang2016commitment}, and recently the primal formulation in~\cite{hua2017convex} has been shown promising in reducing uplift payment~\cite{ELMP2019}. However, a lot of practical features like ramping rate constraints{~\cite{chen2019}}, time-dependent start-up costs, and so on, are not captured in~\cite{hua2017convex}, which makes the formulation not tight for quite a few generators.}%, since a small improvement will result in huge savings in total considering the large total generation amount.}

%There are some works on developing strong formulations for single generator by exploring its special structure.

In this paper, we \revise{follow the spirit of~\cite{schiro2016convex,hua2017convex}} and study the generalization of convex hull pricing problem with the focus on the real-world MISO system, which has more features than those described in the literature, \revise{including time-variant capacity, time-variant ramp-up/-down limit, flexible min-up/-down time limit, max-up time limit constraints, and time-dependent start-up {costs}. Those features are important to improve the model accuracy, while they make it more difficult to describe the convex hull formulation. Therefore, it brings challenges to the primal convex hull formulation framework.} \yu{Besides}, MISO operates one of the world's largest energy markets with more than \$29 billion in annual gross market energy transactions\revise{~\cite{aboutMISO}}. It is required for the market-clearing process to solve large-scale problems with millions of decision variables~\cite{chen2016}. Efficient algorithms are desired to achieve high computational efficiency. The main contributions of this paper are as follows:  
\begin{enumerate}
	\item Based on the specific features of the MISO system{,} including four types of thermal generators, we review convex hull formulations for two types of generators and develop two new convex hull formulations for the other two types of relative more complicated generators. For the class of generators with most complicated physical and operational restrictions, we develop the corresponding integral formulation, which can capture time-variant capacity, time-variant ramp-up/-down limit, flexible min-up/-down time limit, and {max-up time} limit constraints, as well as time-dependent start-up costs altogether, besides traditional physical restrictions. \revise{These features have not been captured altogether by any formulations in the literature.} Our integral formulation could lead to an exact convex hull price for the MISO system by solving a linear program. 
	%based on those convex hull results is derived, which is the most compact convex hull pricing formulation for the MISO system. %to get the exact convex hull pricing and minimum uplift cost.
	\item To solve large-size MISO instances, we develop an innovative iterative algorithm, as well as its variant, to speed up the process to solve the problem. We prove that the iterative algorithm converges as the number of iterations increases. Furthermore, we develop a complementary algorithm to utilize several processors to solve the problem together, which can provide MISO the flexibility to get a better solution within a given time limit.
	\item We test the algorithms on MISO instances. The numerical experiment results show that our proposed algorithms lead to an exact convex hull price and a minimum uplift payment for all of the test instances within an acceptable time limit, \serevise{which is one hour, the time limit given for executing pricing run for MISO's day-ahead market,} for the cases with and without transmission constraints. \end{enumerate}

The remainder of this paper is organized as follows. First, \revise{in Section~\ref{sec:cvp}, we review the convex hull pricing framework and primal formulation method.} Next, in Section~\ref{sec:convex}, we adopt the integral formulations in the literature that can be used in \revise{MISO} convex hull pricing problem. Then, in Section~\ref{sec:exMISO}, we develop two new integral formulations that can capture the special characteristics of MISO instances. After that, in Section~\ref{sec:algorithm}, we describe our innovative iterative algorithm and its variants, as well as the convergence proof of this algorithm. In Section~\ref{sec:num}, we report computational results on MISO instances. \revise{Finally, we conclude our research in Section~\ref{sec:conclusion}.}

\section{Convex Hull Pricing and Uplift Payment Minimization}\label{sec:cvp}
\revise{In this section, we briefly review the convex hull pricing problem as described in~\cite{gribik2007market,schiro2016convex,hua2017convex} and give the corresponding primal formulation.}

\revise{The system optimization problem for a {$|\T|$}-period UCED problem run by an ISO can be abstracted as follows:
\begin{subeqnarray} \label{model:UCED}
	& Z^*_{\tiny\mbox{QIP}} = & \min \sum_{i \in \G} C_i(x_i,u_i,v_i,e_i)  \\
	&\text{st.} & Ex \geq p, \slabel{eqn:sys} \\
	&& (x_i, u_i, v_i, e_i) \in  \Xe_i, \forall i \in \G, \slabel{eqn:single}
\end{subeqnarray}
where $C_i(x_i,u_i,v_i,e_i)$ represents the total cost for generator $i$, including start-up, shut-down, no-load, and generation costs, constraints~\eqref{eqn:sys} represent system-wide constraints including load balance and \threvise{angle-eliminated transmission constraints in terms of the shift factors~\cite{misocost}}, constraints~\eqref{eqn:single} represent the feasible region of each generator $i$.}

\revise{After solving the UCED problem, an ISO obtains the generation amount ($\bar{x}$) and commitment status ($\bar{u},\bar{v},\bar{e}$) of each generator. The ISO also determines the price {based on the dual vector} $\gamma$ (referred to as price $\gamma$ for brevity) for all participants. For a given price $\gamma$ offered by the ISO, the self-scheduling profit maximization problem for each generator $i$ can be described as follows:
\begin{subeqnarray} \label{eq:selfmax}
	&Z^i(\gamma) = & \max \gamma^{\intercal}E_ix_i - C_i(x_i,u_i,v_i,e_i)\\
	&\text{st.} & (x_i, u_i, v_i, e_i) \in  \Xe_i.
\end{subeqnarray}}
\hspace*{2mm}The uplift payment is defined as follows~\cite{schiro2016convex}: 
\begin{eqnarray}
	&U = \sum_{i\in\G} Z^i(\gamma) -& (\gamma^{\intercal}E\bar{x} - \sum_{i\in\G}C_i(\bar{x}_i,\bar{u}_i,\bar{v}_i,\bar{e}_i)) \nonumber \\
	&& +(\gamma^{\intercal}E\bar{x} - \gamma^{\intercal}p),  \label{eq:upr}
\end{eqnarray}
where the first two items represent the difference between the maximum profit obtained through self-scheduling given price $\gamma$ and the profit obtained following ISO's schedule, and the last item is defined as ``Financial Transmission Rights (FTR) uplift" in~\cite{gribik2007market}.

%\threvise{The work in~\cite{gribik2007market} shows that the convex hull price corresponds to the optimal Lagrangian multiplier for the system-wide constraints~\eqref{eqn:sys} of UCED problem. 
	%But the optimal Lagrangian multiplier is hard to be derived for the MIP-based UCED problem. To address this, 
%	Further, the primal formulation method proposed in~\cite{hua2017convex} indicates that the optimal Lagrangian multiplier can be derived by solving a linear programming problem when we have the convex hull description of each single generator and the convex envelope of its cost function. An equivalent formulation is also provided in~\cite{schiro2016convex}.}
\revise{As stated \serevise{in~\cite{hua2017convex}}, the convex hull pricing problem corresponding to problem~\eqref{model:UCED} is presented as the following formulation~\eqref{eq:4}. \serevise{The exact convex hull price \threvise{can be derived as $\bar\gamma^{\intercal} E$ based on} the optimal dual vector $\bar\gamma$ of~\eqref{eqn:sys} in formulation~\eqref{eq:4} \threvise{as introduced in~\cite{misocost} and \cite[Chapter 8.11]{powerlmp} (referred to as convex hull price $\bar\gamma$ for brevity)}}. The minimum uplift payment is guaranteed under the convex hull price $\bar{\gamma}$. Note here, $conv(\Xe_i)$ represents the convex hull formulation for set $\Xe_i$, and $\bar C_i(x_i,u_i,v_i,e_i)$ is the convex envelope for the general cost function of each generator $i$.
\begin{subeqnarray} \label{eq:4}
	& Z^*_{\tiny\mbox{QLP}} = & \min \sum_{i \in \G} \bar C_i(x_i,u_i,v_i,e_i)  \\
	&(\bar\gamma) \quad \text{st.} & \eqref{eqn:sys},  \nonumber \\ %Ex \geq p, \slabel{eqn:sys1}  \\
	&& (x_i, u_i, v_i, e_i) \in  conv(\Xe_i), \forall i \in \G. \slabel{eqn:single1}
\end{subeqnarray}}

\section{Formulations for Several Special Cases} \label{sec:convex}
There are different types of generators in practice within MISO. We first present a traditional 3-bin UC formulation as a base for building our model as follows:
\begin{subeqnarray}\label{eq:1}
\hspace{-0.5in} Z^*_{\tiny\mbox{QIP}} & \hspace{-0.15in} = & \hspace{-0.25in} \min_{f,x,u,v,e,\delta} \sum_{i\in \G}\sum_{t\in \T}\left ( \sum_{s\in S}{\ggy{F_{is}}}\delta_{it}^s + S_ie_{it} + {\ggy{G_{it}}} u_{it} + f_{it}\right ) \slabel{sys:model} \\
\hspace{-0.5in}	\textrm{s.t.} & & ~\eqref{eqn:sys} \nonumber \\
	&&  (f_i,x_i,\delta_i,u_i,v_i,e_i) \in \Xe_i, \forall i \in \G, \slabel{eq:single}
\end{subeqnarray}
where the objective function is to minimize the total cost, {in which the four items represent start-up, shut-down, no-load, and generation costs. {The feasible region $\Xe_i$ of generator $i$ can be expressed as follows:}}
	%The feasible set of a single generator $\Xe_i$ can be expressed as follows~\cite{}:
	\begin{subeqnarray} \label{eq:sg}
	&& \hspace{-0.3in}\Xe_i = \{f_i,x_i,\delta_i,u_i,v_i,e_i \in \hspace{-0.03in}\R^{\tiny |\hspace{-0.02in}\T\hspace{-0.02in}|} \hspace{-0.06in} \times\hspace{-0.03in} \R^{\tiny |\hspace{-0.02in}\T\hspace{-0.02in}|} \hspace{-0.06in}\times \hspace{-0.03in} \hspace{-0.03in} \R^{\tiny |\hspace{-0.02in}\T\hspace{-0.02in}|} \hspace{-0.06in} \times \hspace{-0.03in} \B^{\tiny |\hspace{-0.02in}\T\hspace{-0.02in}|} \hspace{-0.06in}\times \hspace{-0.03in}\B^{\tiny |\hspace{-0.02in}\T\hspace{-0.02in}|\hspace{-0.02in}-\hspace{-0.02in}1} \hspace{-0.06in} \times \hspace{-0.03in} \B^{\tiny|\T|\hspace{-0.02in}-\hspace{-0.02in}1}\hspace{-0.03in}|\nonumber\\
	&& f_{it} \geq {\ggy{H_{it}^{k}}} x_{it} - {\ggy{Q_{it}^{k}}} u_{it}, \forall k \in [1,N]_\Z, \forall t \in \T,\slabel{eq:piece}\\
	&& {\ggy{\delta_{it}^s}} \leq \sum_{j = \underline{T}_i^{s}}^{{\overline{T_i}^s}} e_{i(t-j)}, \forall s \in \mathcal{S}/\{\text{c}\}, t \in \T, \slabel{eq:start1}\\
	&& \sum_{s\in \mathcal{S}} {\ggy{\delta_{it}^s}} = v_{it}, t \in \T,\slabel{eq:start2}\\
	&& u_{it} - u_{i(t-1)} = v_{it} - e_{it},\forall t \in [2,|\T|]_\mathbb{Z}, \slabel{eq:logic}\\
	&& \sum_{j=t-L_i+1}^{ t} v_{ij} \leq u_{it},  \forall t\in [L_i+1,|\T|]_\mathbb{Z}, \slabel{eq:minup}\\
	&& \sum_{j=t-\ell_i +1}^{ t} v_{ij} \leq 1-u_{i(t-\ell_i)}, \forall t \in [\ell_i+1, |\T|]_\mathbb{Z},\slabel{eq:mindown}\\
	&& x_{it} \geq  \pl_i u_{it},\forall t \in \T, \slabel{eq:capl}\\
	&& x_{it} \leq \pb_i u_{it}, \forall t \in \T, \slabel{eq:capu}\\
	&& x_{it}-x_{i(t-1)} \leq \V_{i}^u u_{i(t-1)} + \Vb_{i}^{u} v_{it}, \forall t \in \T, \slabel{eq:rampup}\\
	&& x_{i(t-1)}-x_{it} \leq \V_{i}^e u_{it} + \Vb_{i}^{e} e_{it}, \forall t \in \T\}, \slabel{eq:rampdown}
	\end{subeqnarray}
where $\pl_i$,$\pb_i$,$V^u_i$,$\Vb_i^u$,$V^e_i$,$\Vb^e_i$ represent time-invariant parameters. Constraints~\eqref{eq:piece} represent the piecewise linear approximation of the convex cost functions. Constraints~\eqref{eq:start1} and~\eqref{eq:start2} use indicator variables to represent each start-up type, which depend on the number of time periods the generator has been off before it was started-up. Constraints~\eqref{eq:logic} represent the logic relationships. Constraints~\eqref{eq:minup} to \eqref{eq:rampdown} represent the min-up/-down time restrictions (i.e.,~\eqref{eq:minup}-\eqref{eq:mindown}), generation lower and upper bounds (i.e.,~\eqref{eq:capl}-\eqref{eq:capu}), and ramp-up/-down rates (i.e.,~\eqref{eq:rampup}-\eqref{eq:rampdown}). 

Now we introduce convex hull results for two types of generators in MISO below.
\subsubsection{\rere{The set of generators with constant start-up costs and without ramping constraints (labeled as set $\G_1$)}}
Among various types of generators, low capacity fast-start gas generators are relatively easy to model in terms of physical constraints. They can start-up quickly\yu{,} and the ramping capability is larger than the gap between upper and lower generation limits. Also, the start-up costs are constant and not dependent on the down-time periods before \yu{the} start-up. Thus, constraints~\eqref{eq:start1},~\eqref{eq:start2},~\eqref{eq:rampup}, and~\eqref{eq:rampdown} are redundant and~\eqref{eq:sg} can be simplified as follows:
		\begin{eqnarray}\label{mod:1}
		Z^*_{\tiny\mbox{QIP}} & =\hspace{0.1in} & \hspace{-0.2in}\min_{x,u,v,e} \sum_{i\in \G}\sum_{t\in \T}(S'_i v_{it}  + S_ie_{it} + {\ggy{G_{it}}} u_{it} + f_{it}) \label{sys:model2} \\
		\textrm{s.t.}& & ~\eqref{eqn:sys}, (f_i,x_i,u_i,v_i,e_i) \in {\ggy{\Xe^1_i}}, \forall i \in \G_1.\nonumber\\
		{\ggy{\Xe^1_i}} & =\hspace{0.1in}&\hspace{-0.2in} \{f_i,x_i,u_i,v_i,e_i \in \R^{\tiny |\T|} \hspace{-0.05in}\times \hspace{-0.05in}\R^{\tiny |\T|}\hspace{-0.05in} \times\hspace{-0.05in} \B^{\tiny |\T|} \hspace{-0.05in}\times\hspace{-0.05in} \B^{|\T|-1} \hspace{-0.05in}\times\hspace{-0.05in} \B^{|\T|-1}| \nonumber\\
		&& ~\eqref{eq:piece},~\eqref{eq:logic}-\eqref{eq:capu}\}.  
		\end{eqnarray}
The convex hull (as described in~\cite{hua2017convex} and~\cite{Gentile2017}) for ${\ggy{\Xe^1_i}}$ is
\begin{eqnarray}\label{mod:2}
{\ggy{\De^1_i}} & =\hspace{0.1in}&\hspace{-0.2in} \{f_i,x_i,u_i,v_i,e_i \in \R^{\tiny |\T|} \hspace{-0.05in}\times \hspace{-0.05in}\R^{\tiny|\T|}\hspace{-0.05in} \times\hspace{-0.05in} \R^{\tiny|\T|} \hspace{-0.05in}\times\hspace{-0.05in} \R^{|\T|-1} \hspace{-0.05in}\times\hspace{-0.05in} \R^{|\T|-1}| \nonumber\\
		&& ~\eqref{eq:piece},~\eqref{eq:logic}-\eqref{eq:capu},\nonumber\\
		&& {v_{i}} \geq 0, {e_{i}} \geq 0 \slabel{eq:vw}\}.  
\end{eqnarray}
%In later Section~\ref{subsec:misogene}, we present the convex hull results that have accommodated the specific physical features of the generators for the MISO instances.
		
\subsubsection{\rere{The set of generators with constant start-up costs and start-up ramping constraints (labeled as set $\G_2$)}}
For another group of gas-fired generators with less restrictive physical constraints, the ramping rates during the stable region are larger than the gap between the upper and lower generation bounds. But the start-up ramping rates ($\Vb_i^u$) are binding. For this case, the convex hull description is (as described in \cite{Gentile2017}) 
\begin{eqnarray}\label{mod:3}
		{\ggy{\De^2_i}} & =\hspace{0.1in}&\hspace{-0.2in} \{f_i,x_i,u_i,v_i,e_i \in \R^{\tiny |\T|} \hspace{-0.05in}\times \hspace{-0.05in}\R^{\tiny |\T|}\hspace{-0.05in} \times\hspace{-0.05in} \R^{\tiny |\T|} \hspace{-0.05in}\times\hspace{-0.05in} \R^{|\T|-1} \hspace{-0.05in}\times\hspace{-0.05in} \R^{|\T|-1}| \nonumber\\
		&& ~\eqref{eq:piece},~\eqref{eq:logic}-\eqref{eq:capu},~\eqref{eq:vw}, \nonumber \\
		&& \hspace{-0.2in}x_{it} \leq  \pb_i u_{it} + (\Vb^u_{i} - \pb_i)v_{it}, \forall t \in [2,|\T|]_\Z \slabel{eq:strampup}\}.
\end{eqnarray}

\section{Formulations for the General MISO Instances} \label{sec:exMISO}
For the generators in MISO, besides special generators $\G_1$ and $\G_2$ as described in~\ref{sec:convex}, there are several generator features which reflect the market needs and capture more details in practice. \revise{In this section, we first introduce three new features in general MISO instances. Then based on generator types, we derive two convex hull descriptions capturing these features.} 

\subsection{\revise{Features in the general MISO instances}}

\subsubsection{\revise{Max-up time limit}}\label{case:a1}
For some generators in the MISO market, there are restrictions on maximum time periods that the generator can stay online because of machine deterioration. To accommodate this, the max-up time constraints can be described as follows:
\begin{equation}
\sum_{j=t+1}^{t+\overline{L}_i} v_{ij} \geq u_{i(t+\overline{L}_i)},\forall t \in \T \slabel{eq:maxon}
\end{equation}

\subsubsection{\revise{Flexible min-up/-down time limit}}\label{case:a2}
In MISO, the min-up/-down time limit is set to be time-variant to resolve offer data conflicts. For example, participants may submit a must-run offer for a generator for hours $1-5$ and $10-24$ with the min-down time limit as $6$ hours. This will force UCED to commit this generator between $6$ and $9$ even if it is \yu{costly}. So MISO developed a set of rules to ignore min-up/-down time limit if there are such conflicts. In this example, the min-down time limit is relaxed to be $1$ between hours $6$ and $9$ so that the generator will not be committed if it is \yu{costly}. \ggy{It will force market participants to submit proper offers and prevent market manipulation. If they do want to run through, they should submit \yu{a} must-run offer for all hours. If they do want MISO to determine on/off in between, they should have \yu{at least} $6$ hours in between.} To accommodate this, the refined min-up/-down time constraints in the $3$-bin formulation (such as $\G_1$ and $\G_2$) can be described as follows (generator index $i$ is omitted for brevity):
\begin{eqnarray}
&& \sum_{j=t-L+1}^{t} \kappa_{t} v_{j} \leq u_{t},  \forall t\in [L+1,|\T|]_\mathbb{Z}, \label{eq:misomu}\\
&& \sum_{j=t-\ell +1}^{\tiny |\mathcal{T}|} \varpi_{t} v_{j} \leq 1-u_{t-\ell},   \forall t \in [\ell+1, |\T|]_\mathbb{Z}.\label{eq:misomd}
\end{eqnarray}

\subsubsection{\revise{Time-variant parameters}}\label{case:a3}
\ggy{In MISO, market participants are allowed to offer capacity and ramp rates varying by the hour.} The time-variant parameters make the convex hull more complicated and have rarely been studied in the literature.

\subsection{\revise{Convex hull results for general MISO instances}}

\subsubsection{\revise{The set of generators in $\G_2$ with max-up time limit constraints described in~\ref{case:a1} (labeled as set $\G_3$)}}{For this type of generators, we derive and prove the convex hull description in Theorem~\ref{thm1}. }
\begin{theorem} \label{thm1}
	The convex hull description for the 3-bin model with max-up time restriction can be described as follows:
	\begin{eqnarray} \label{mod:4}
	{\ggy{\De^3_i}} & =\hspace{0.1in}&\hspace{-0.2in} \{f_i,x_i,u_i,v_i,e_i \in \R^{\tiny |\T|} \hspace{-0.05in}\times \hspace{-0.05in}\R^{\tiny |\T|}\hspace{-0.05in} \times\hspace{-0.05in} \R^{\tiny |\T|} \hspace{-0.05in}\times\hspace{-0.05in} \R^{|\T|-1} \hspace{-0.05in}\times\hspace{-0.05in} \R^{|\T|-1}| \nonumber\\
	&& ~\eqref{eq:piece},~\eqref{eq:logic}-\eqref{eq:capu},~\eqref{eq:vw}-\eqref{eq:maxon} \}. 
	\end{eqnarray}
\end{theorem}	
\begin{proof}
	Based on Proposition $2$ in~\cite{Queyranne2017}, for the min-up/-down time only polytope (without ramping constraints) with \ggy{max-up time} restrictions, the convex hull description of the feasible binary variables ($u,v, e$) is $\{u,v,e \in \R^{\tiny |\T|} \times \R^{\tiny |\T|} \times \R^{\tiny |\T|}|~\eqref{eq:logic}-\eqref{eq:mindown},\eqref{eq:vw},\eqref{eq:maxon}\}$. Then, based on Lemma $4$ in~\cite{Gentile2017}, the addition of capacity constraints~\eqref{eq:capl}-\eqref{eq:capu}, start-up ramping constraints~\eqref{eq:strampup}, and linear function~\eqref{eq:piece} does not affect integrality. Thus, the conclusion holds.
\end{proof}

\subsubsection{\revise{The set of generators with all the additional restrictions described in \ref{case:a1}-\ref{case:a3}(labeled as set $\G_4$)}} \label{subsub:4}
%\revise{For this type of generators, we derive a convex hull description by using a primal-dual approach motived by a revised dynamic programming algorithm. The intuitive idea is that the UCED problem is a sequential decision-making problem, given the net generation cost for all possible on-intervals, we can build a revised dynamic programming framework to get the optimal solution and derive an equivalent linear program. Further, we use the primal-dual approach to incorporate the economic dispatch problem, which is used to calculate the net generation cost, into the proposed linear program. Our convex hull description, as well as the integral formulation, is obtained by taking the dual of the integrated linear program and replacing the equivalent variables. }

%\revise{For this type of generators, we derive an convex hull description by using a network flow structure to capture the on/off status of the generator. The constraints related to binary variables, such as min-up/-down time limit and max-up time limit, can be represented by designing the arcs in the network flow, which ensures the binary values achieved at the extreme points. Constraints for the generation amount (i.e., capacity, ramping rate limit) are attached to each on-interval, therefore we define variables $q_{tk}^s$ corresponding to each on-interval to represent the generation amount at time $s$ if the generator is online during the interval $[t,k]_\Z$ ($s \in [t,k]_\Z$). }
	
\revise{For this type of general generators, the model is built in a high-dimensional space by introducing binary variables $y_{tk}$ ($z_{tk}$) to indicate a generator is ``ON'' (``OFF'') from time periods $t$ to $k$, in order to keep track of its ``ON'' and ``OFF'' intervals. Accordingly, a new decision variable $q^s_{tk}$, corresponding to each ``ON'' interval $[t, k]_\Z$, is introduced to represent the generation amount at time $s, \ t \leq s \leq k$.} Now, we describe the integral formulation (referred \yu{to} as EUC formulation) and give the detailed illustrations in the proof of Theorem~\ref{thm2}. We assume the generator has been initially on for $s_0$ time periods before time $1$. Thus, the generator cannot shut down until time $t_0+1$, with $t_0 = [L-s_0]^+$, due to min-up constraints.  The detailed formulation is shown as follows:
	\begin{subeqnarray} \label{model:LP_D3}
	&\hspace{-0.5in}\min  & \hspace{-0.35in} \sum_{k \in \overline{\K},k <|\T|} \hspace{-0.15in} \textcolor{black}{S}(\ggy{k+s_0})\textcolor{black}{w_k} +\hspace{-0.1in} \sum_{tk \in \overline{\T\K^2},k<|\T|}\hspace{-0.25in} \textcolor{black}{S}(k-t+1) \textcolor{black}{y_{tk}} \nonumber\\
	&& +\hspace{-0.1in} \sum_{kt \in \overline{\K\T}}\hspace{-0.1in} \textcolor{black}{S'}(t-k-1) \textcolor{black}{z_{kt}} +\hspace{-0.1in} \sum_{tk \in \overline{\T\K} }  \sum_{s=t}^{k} \phi_{tk}^s  \slabel{eqn:LP_D3obj}\\
	%&& \ \sum_{k=t_0}^{T-\ell-1}\sum_{t=k+\ell+1}^{T} \textcolor{black}{S} \textcolor{black}{z_{kt}} + \sum_{tk \in \T \K }  \sum_{s=t}^{k} \phi_{tk}^s  \slabel{eqn:LP_D1obj}\\
	&\hspace{-0.5in}\mbox{s.t.}  & \sum_{k \in \overline{\K}} \textcolor{black}{w_k} = 1, \slabel{eqn:LP_D21}\\
	&& \hspace{-0.5in}-\textcolor{black}{w_{t'}1_{\{t'\hspace{-0.015in}\in \hspace{-0.015in}\overline{\K}\}}} \hspace{-0.05in} + \hspace{-0.05in} \hspace{-0.2in}\sum_{tk \in \overline{\K\T},t=t'}\hspace{-0.15in} \textcolor{black}{z_{t'k}} - \hspace{-0.22in}\sum_{kt \in \overline{\T\K^2},t=t'}\hspace{-0.2in}\textcolor{black}{y_{kt'}} \hspace{-0.05in} +\hspace{-0.03in} \theta_{t'} \hspace{-0.05in} = \hspace{-0.05in} 0,\hspace{-0.05in} \forall t'\hspace{-0.05in} \in\hspace{-0.05in} [t_0,\hspace{-0.03in} |\T|\hspace{-0.05in}-\hspace{-0.03in}1]_{\Z},\slabel{eqn:LP_D32}\\
	&& \hspace{-0.5in}	\sum_{tk \in \overline{\T\K^2},t=t'}\hspace{-0.2in}\textcolor{black}{y_{t'k}} -\hspace{-0.1in} \sum_{kt \in \overline{\K\T},t=t'}\hspace{-0.2in} \textcolor{black}{z_{kt'}}= 0, \forall t' \in [t_0+\ell_t+1, |\T|]_{\Z},\slabel{eqn:LP_D33}\\
	&& \hspace{-0.5in}\pl_s\textcolor{black}{w_k} \leq q_{tk}^s \leq \pb_s\textcolor{black}{w_k}, \ \forall s \in [t,k]_\Z, \forall tk \in \overline{\T\K^1},  \slabel{eqn:LP_D351}\\
	&& \hspace{-0.5in}\pl_s\textcolor{black}{y_{tk}} \leq q_{tk}^s \leq \pb_s\textcolor{black}{y_{tk}}, \ \forall s \in [t,k]_{\Z}, \forall tk \in \overline{\T\K^2}, \slabel{eqn:LP_D352}\\
	&& \hspace{-0.5in}q_{tk}^{\tiny t} \leq \overline{V}^u_{\tiny t}\textcolor{black}{y_{tk}}, \ \forall tk \in \overline{\T\K^2}, \slabel{eqn:LP_D36}\\
	&& \hspace{-0.5in}q_{tk}^k \leq \Vb^e_k\textcolor{black}{w_k}, \ \forall tk \in \overline{\T\K^1}, k\leq |\T|-1 \slabel{eqn:LP_D361}\\
	&& \hspace{-0.5in}q_{tk}^k \leq \Vb^e_k\textcolor{black}{y_{tk}}, \ \forall tk \in \overline{\T\K^2}, k\leq |\T|-1 \slabel{eqn:LP_D362}\\
	&& \hspace{-0.5in}q_{tk}^{s-1} - q_{tk}^s \leq\V^e_s \textcolor{black}{w_k}, \ q_{tk}^{s} - q_{tk}^{s-1} \leq\V^u_s \textcolor{black}{w_k}, \nonumber\\
	&& \hspace{-0.5in} \forall s \in [t+1,k]_{\Z}, \forall tk \in \overline{\T\K^1}, \slabel{eqn:LP_D37}\\
	&& \hspace{-0.5in}q_{tk}^{s-1} - q_{tk}^s \leq\V^e_s \textcolor{black}{y_{tk}}, \  q_{tk}^{s} - q_{tk}^{s-1} \leq\V^u_s \textcolor{black}{y_{tk}}, \nonumber\\
	&& \hspace{-0.5in} \forall s \in [t+1,k]_{\Z}, \forall tk \in \overline{\T\K^2}, \slabel{eqn:LP_D38}\\
	&& \hspace{-0.5in}\phi_{tk}^s\hspace{-0.03in} -\hspace{-0.03in} a_j q_{tk}^{s}\hspace{-0.03in} \geq\hspace{-0.03in} {b_j} \textcolor{black}{w_k},\hspace{-0.03in} \forall s \hspace{-0.03in} \in\hspace{-0.03in} [t,k]_{\Z}, j\hspace{-0.03in} \in\hspace{-0.03in}  [1,N]_{\Z}, \forall tk \in \overline{\T\K^1}, \slabel{eqn:LP_D391}\\
	&& \hspace{-0.5in}\phi_{tk}^s \hspace{-0.03in}-  \hspace{-0.03in}a_j q_{tk}^{s}\hspace{-0.03in} \geq\hspace{-0.03in} {b_j} \textcolor{black}{y_{tk}}, \hspace{-0.03in}\forall s  \hspace{-0.03in}\in\hspace{-0.03in} [t,k]_{\Z}, j  \hspace{-0.03in}\in \hspace{-0.03in} [1,N]_{\Z}, \forall tk \in \overline{\T\K^2}, \slabel{eqn:LP_D392}\\
	&& \hspace{-0.5in}\textcolor{black}{w}, \textcolor{black}{z},y \geq 0, \theta_t \geq 0, \forall t \in [t_0,|\T|-\ell_t-1]_\Z, \slabel{eqn:LP_D393}
	%&& \hspace{-0.5in}\forall k \in [t+L-1, T-1]_{\Z}, \forall t \in [t_0+\ell+1, T-L]_{\Z}. \slabel{eqn:LP_D292}
\end{subeqnarray}
where \revise{the objective function is to minimize the total cost, including shut-down, start-up and generation costs, constraints~\eqref{eqn:LP_D21}-\eqref{eqn:LP_D33} keep track of the generator's ``ON'' and ``OFF'' statuses ($1_{\{t'\hspace{-0.015in}\in \hspace{-0.015in}\overline{\K}\}}$ indicates that $w_{t'}$ is included in~\eqref{eqn:LP_D32} only when $t'$ is in set $\overline{\K}$), and constraints~\eqref{eqn:LP_D351}-\eqref{eqn:LP_D392} describe the economic dispatch restrictions in the ``ON'' interval, including the generation upper and lower bounds~\eqref{eqn:LP_D351}-\eqref{eqn:LP_D352}, start-up and shut-down ramping restrictions~\eqref{eqn:LP_D36}-\eqref{eqn:LP_D362}, general ramp-up and ramp-down restrictions~\eqref{eqn:LP_D37}-\eqref{eqn:LP_D38}, and finally the piecewise linear convex function ${f_{tk}^s(q_{tk}^{s})}$ representation~\eqref{eqn:LP_D391}-\eqref{eqn:LP_D392}.} \revise{Meanwhile, considering the new features in terms of max-up and flexible min-up/-down time limit, we define  $\overline{\K} = [t_0,\min\{\overline{L}-s_0,|\T|\}]_\Z$ as the subscript set of variable $w$, $\overline{\T\K}$ ($\overline{\K\T}$) as all possible on-intervals (off-intervals) satisfying min-up and max-up (min-down) time restrictions.} More specifically, we let $\overline{\T\K} = \overline{\T\K^1} \cup \overline{\T\K^2}$, where $\overline{\T\K^1}$ collects all possible first ``ON'' intervals $[t, k]_{\Z}$ with $ t = 1$ and $k \in [t_0+1,\min\{\overline{L}-s_0,|\T|\}]_\Z$ and $\overline{\T\K^2}$ collects all following ``ON'' intervals $[t, k]_{\Z}$, after a shut-down has happened, with $t \in [t_0+\ell_t+1, |\T|]_{\Z}$ and $k \in [\min\{t+L_t-1,|\T|\}, \min\{t+\overline{L}-1,|\T|\}]_{\Z}$. \revise{In this expression, we have i) $L_t = L$ if $\kappa_{t}=1$ and $L_t=1$ otherwise, and ii) $\ell_t = \ell$ if $ \varpi_k = 1$ and $\ell_t = 1$ otherwise to reflect the flexible min-up/-down time limit as described earlier in Section~\ref{case:a2}.} Similarly, we have $ \overline{\K\T}$ collecting all ``OFF'' intervals $[k, t]_{\Z}$ with $k \in [t_0, |\T|-\ell_t -1]_{\Z}$ and $t \in [k+\ell_t+1, |\T|]_{\Z}$. \revise{The new feature in terms of time-variant parameters is considered by allowing parameters $\overline{P}_s,\underline{P}_s,\overline{V}_t^u,\overline{V}_k^e, V_s^e$ to be dynamic in time.}

\begin{theorem}\label{thm2}
The convex hull description for the general single generator MISO UC is as follows:
\begin{eqnarray} \label{mod:5}
		{\ggy{\De^4_i}} & =\hspace{0.1in}&\hspace{-0.2in} \{w_i,y_i,z_i,\theta_i, q_i, \phi_i | ~\eqref{eqn:LP_D21}-\eqref{eqn:LP_D393}\},  
	\end{eqnarray}
\revise{and there exists an optimal solution to~\eqref{model:LP_D3} binary with respect to decision variables $w$, $y$, and $z$.}
\end{theorem}	
\begin{proof}
\revise{The detailed proof is provided in Appendix~\ref{app:1}. } 
\end{proof}
Finally, \revise{it can be observed that the convex hull description for the initial-off generators can be obtained similarly.} 
\section{Iterative Algorithms for Convex-hull Pricing for MISO Instances} \label{sec:algorithm}

%\rere{Our work in~\cite{} applies the integral formulation to the convex hull pricing problem and solves this problem by linear programming}. 
Based on the above convex hull descriptions~\eqref{mod:2},~\eqref{mod:3},~\eqref{mod:4}, and~\eqref{mod:5}, we derive the general convex hull pricing formulation ($P$) as follows. For notation brevity, we let $\G_4 = \G/(\G_1 \cup \G_2 \cup \G_3)$ represent all unclassified generators using formulation~\eqref{model:LP_D3}. Since we have the convex hull descriptions for each type of generators and convex envelope for all cost functions, we can provide the exact convex hull price and minimize the uplift payment by solving the following linear program:	
\begin{subeqnarray}
	\hspace{-0.3in}(P): Z^*_{\ggy{\tiny\mbox{LPO}}} & =\hspace{0.1in} & \hspace{-0.5in}\min_{f,x,u,v,e,w,y,z,\theta,q,\phi} \sum_{i\in \G_1\cup\G_2\cup\G_3}\sum_{t\in \T}g_{it} +\sum_{i \in \G_4} g'_{it} \label{sys:total1} \\
	\textrm{s.t.}& & ~\eqref{eqn:sys}, (f_i,x_i,u_i,v_i,e_i) \in {\ggy{\De^1_i}}, \forall i \in \G_1,\nonumber\\
	&& (f_i,x_i,u_i,v_i,e_i) \in {\ggy{\De^2_i}}, \forall i \in \G_2,\nonumber\\
	&& (f_i,x_i,u_i,v_i,e_i) \in {\ggy{\De^3_i}}, \forall i \in \G_3,\nonumber\\
	&& (w_i,y_i,z_i,\theta_i, q_i, \phi_i) \in {\ggy{\De^4_i}}, \forall i \in \G_4,\nonumber
	%\Xe'_i & =\hspace{0.1in}&\hspace{-0.2in} \{f_i,x_i,u_i,v_i,w_i \in \R^{\tiny T} \hspace{-0.05in}\times \hspace{-0.05in}\R^{\tiny T}\hspace{-0.05in} \times\hspace{-0.05in} \B^{\tiny T} \hspace{-0.05in}\times\hspace{-0.05in} \B^{T-1} \hspace{-0.05in}\times\hspace{-0.05in} \B^{T-1}| \nonumber\\
	%&& ~\eqref{eq:piece},~\eqref{eq:logic}-\eqref{eq:rampdown}.\}  
\end{subeqnarray}
where the cost function $g_{it}$ of each generator $i$ in $\G_1\cup\G_2\cup \G_3$ can be expressed as
\begin{equation}
g_{it} = S'_i v_{it}+ S_ie_{it} + {\ggy{G_{it}}} u_{it} + f_{it}     \label{objori1}
\end{equation} 
and the cost function $g'_{it}$ of each generator $i$ in $\G_4$ can be expressed as
\begin{eqnarray}
g'_{it}& = &\hspace{-0.1in}\hspace{-0.1in}\sum_{k \in \overline{\K},k <|\T|} \hspace{-0.1in} \textcolor{black}{S_i}(k+s_0)\textcolor{black}{w_{ik}} +\hspace{-0.15in} \sum_{tk \in \overline{\T\K^2_i},k<|\T|} \hspace{-0.15in}\textcolor{black}{S_i}(k-t+1) \textcolor{black}{y_{itk}} \nonumber\\
	&& \hspace{-0.15in} +\sum_{kt \in \overline{\K\T}}\textcolor{black}{S'_i}(t-k-1) \textcolor{black}{z_{ikt}} + \sum_{tk \in \overline{\T\K} }  \sum_{s=t}^{k} \phi_{itk}^s. \label{objeuc}
\end{eqnarray}
	
The formulation ($P$) is sufficient to solve the convex hull pricing problem. However, a large number of variables and constraints in the EUC formulation (i.e.,~\eqref{model:LP_D3}) increase the computational complexity and lead to a long solving time when the number of generators in $\G_4$ is large. To solve the large-scale problem ($P$), we develop an iterative algorithm, in which a relaxation problem \revise{where $\De^4_i$ is replaced by $\De^3_i$ and $g'_{it}$ is replaced by $g_{it}$ for each generator $i \in \G_4$ in ($P$)}  is solved in the first step. Then the constraints \revise{$\De^4_i$} in EUC formulation are added gradually when needed to tighten the relaxation. \revise{Meanwhile, the objective function $g_{it}$ is replaced by $g'_{it}$ when $\De_i^4$ is added back}. Our approach can provide a very tight approximation of the convex hull price\yu{,} with most of the cases converging at the optimal solution. Meanwhile the algorithm terminates in a short time. 
\subsection{\rere{Methodology background}}
Before describing the detailed algorithms, we first introduce Lemma~\ref{lemma1} and Theorem~\ref{thm3} to provide a theoretical foundation showing that the uplift payment amount will decrease and converge under our algorithm. For notation brevity, we use $x \in \R^n$ to represent all variables in~\eqref{eq:1}, which contains binary variable set $x_1$ and continuous variable set $x_2$. The traditional 3-bin MILP UC formulation~\eqref{eq:1} can be abstracted as follows: 
	\begin{eqnarray}
	Z^*_{\tiny\mbox{QIP}}= \min\{c^{\intercal}x| Ex \geq  p, x\in \Xe\} ,\label{simeq1}\\
	\Xe =  \{x_1 \in \B^{n_1}, x_2 \in \R^{n_2}|Ax \leq b\}. \nonumber
	\end{eqnarray}
	%where  $ $.
	\begin{lemma}\label{lemma1}
		For $A \in \R ^{m \times n}$, $A' \in \R ^{m \times n}$, $\revise{E} \in \R^{p \times n}$, $b\in\R^m$, $b'\in\R^m$, $\revise{p}\in \R^p$ and $c \in \R^{n}$, we consider the integer optimization problem~\eqref{simeq1}, its tightened linear programming relaxation problem $(P_\text{C})$ in which $A'x \leq b'$ dominates $Ax \leq b$, and the Lagrangian relaxation $\De_\text{C}$ corresponding to a \dualvector{dual vector} $\gamma$ as follows:
		\begin{eqnarray}
		(P_{\text{C}}): & \hspace{-0.4in}\ggy{Z_{\text{\tiny C}}} = \min\{c^{\intercal}x| A'x \leq b', \revise{Ex \geq p}, x\in \R^n\}, \label{pt}\\
		(\De_{\text{C}}):& \hspace{-0.1in}\ggy{Z_{\text{\tiny C}}}(\gamma) \hspace{-0.05in} = \min \{c^{\intercal}\hspace{-0.03in}x\hspace{-0.03in}+\hspace{-0.03in}\gamma^{\intercal}\hspace{-0.03in} (p\hspace{-0.03in}-\hspace{-0.03in}Ex)| A'x \leq b', x \in \R^n\}. \label{lptr}
		\end{eqnarray}
Given an optimal \dualvector{dual vector} $\bar \gamma$ for constraints $Ex\geq p$ in~\eqref{pt}, if% and its tighter linear programming relaxation problem by adding valid cuts:
		%\begin{eqnarray*}
		%(LP): &z_L = \min\{c^{\tiny T}x| Ax \leq b, Dx = d, x\in \R^n\}
		%(LP_t): &z_{TL} = \min\{c^{\tiny T}x| A'x \leq b', Dx = d, x\in \R^n\}
		%\end{eqnarray*}
		%where, $A'x \leq b'$ dominates $Ax \leq b$. Their Lagrangian relaxation can be written as follows:
		%\begin{eqnarray*}
		%(L_1):&z_I(\tilde{\gamma}) = \min\{c^{\tiny T}x + \tilde{\gamma}^{\tiny T}(d-Dx)| Ax \leq b, x\in \B^n\} \nonumber\\
		%(L):&z_L(\gamma) = \min\{c^{\tiny T}x+\gamma^{\tiny T}(d-Dx)| Ax \leq b, x\in \R^n\} 
		%(L_3):&z_{TL}(\gamma')=\min\{c^{\tiny T}x+\gamma'^{\tiny T}(d-Dx)| A'x \leq b', x\in \R^n\}
		%\end{eqnarray*}
		%Further, we present the Lagrangian dual problem:
		%\begin{eqnarray*}
		%&(\De_{I}):& w_{I} =  \max \{z_I(\tilde{\gamma}), \tilde{\gamma}\in\R^p  \}\\
		%&(\De_{L}):& w_{L} =  \max \{z_L(\gamma), \gamma\in\R^p  \}
		%&(\De_{L_t}):& w_{L_t} =  \max \{z_{TL}(\gamma'), \gamma'\in\R^p  \}
		%\end{eqnarray*}
		
		\begin{enumerate}[(i)]
			\item $x^*=\{x_1^*,x_2^*\}$ is an optimal solution to problem $\ggy{Z_{\text{\tiny C}}}(\bar \gamma)$,\label{condition1}
			\item $x_1^*$ are all binaries, \label{condition2}
		\end{enumerate}
		then the uplift payment given price $\bar \gamma$ can be calculated as $U = Z^*_{\tiny\mbox{QIP}}-\ggy{Z_{\text{\tiny C}}}(\bar \gamma)$ {($\boxtimes$)}.
	\end{lemma}
	\begin{proof}
		Based on the definition of uplift payment in~\eqref{eq:upr}, the profit following the ISO's schedule can be calculated as $P_{\tiny\text{ISO}} = \revise{\bar\gamma^{\intercal}E\bar{x}} - Z^*_{\tiny\mbox{QIP}}$ {($\diamondsuit$), which is the second item in~\eqref{eq:upr}}, and the maximum profit $\sum_{i\in\G} Z^i(\bar\gamma)$ obtained through self-scheduling {problem~\eqref{eq:selfmax}} given price $\bar \gamma$ can be derived by solving the following problem{, where $c^{\intercal}x$ is an abstract form of $ C_i(x_i,u_i,v_i,e_i)$ in~\eqref{eq:selfmax}}:
		\begin{equation}
		\hspace{-0.02in}\sum_{i\in\G} Z^i(\bar\gamma) = -\min\{c^{\intercal}x-\bar\gamma^{\intercal}Ex |A'x \leq b', x_1 \in \B^{n_1}\}.~\label{eq:self}
		\end{equation} 
		
{Based on formulations~\eqref{lptr} and~\eqref{eq:self}}, it is clear that $\ggy{ -\sum_{i\in\G} Z^i(\bar\gamma) \geq Z_{\text{\tiny C}}}(\bar\gamma)-\bar\gamma^{\intercal}p$ {($\bigstar$)}, since $x_1$ is relaxed to be continuous between $0$ and $1$ in~\eqref{lptr}. Meanwhile, when conditions~\eqref{condition1} and~\eqref{condition2} hold, $x^*$ is a feasible solution to~\eqref{eq:self}. Thus, $ -\sum_{i\in\G} Z^i(\bar\gamma) \leq \ggy{Z_{\text{\tiny C}}}(\bar\gamma)-\bar\gamma^{\intercal}p$ {($\clubsuit$)}. {Combining ($\bigstar$) and ($\clubsuit$),} we have $\sum_{i\in\G} Z^i(\bar\gamma) = -(\ggy{Z_{\text{\tiny C}}}(\bar\gamma)-\bar\gamma^{\intercal}p)$ {($\divideontimes$)}.
		
Therefore, $U=\sum_{i\in\G} Z^i(\bar\gamma)-P_{\tiny\text{ISO}}\revise{+\bar\gamma^{\intercal}E\bar{x}-\bar\gamma^{\intercal}p} =-(\ggy{Z_{\text{\tiny C}}}(\bar\gamma)-\bar\gamma^{\intercal}p)-(\revise{\bar\gamma^{\intercal}E\bar{x}} - Z^*_{\tiny\mbox{QIP}})\revise{+\bar\gamma^{\intercal}E\bar{x}-\bar\gamma^{\intercal}p} = Z^*_{\tiny\mbox{QIP}}-\ggy{Z_{\text{\tiny C}}}(\bar \gamma)$, {where the first equation holds due to~\eqref{eq:upr}, the second equation holds based on ($\divideontimes$) and ($\diamondsuit$), and the last one is directly derived by eliminating the same items}.
\end{proof}

\begin{remark}
\serevise{Note here that the prices $\bar\gamma$ are calculated by solving a linear program\yu{,} and thus, they are not linked to the optimal integer solutions in the original problem.}%we can use the best upper bound of  to replace $Z^*_{\tiny\mbox{QIP}}$ in Lemma 1 to calculate $U$.}
\end{remark}
\begin{remark}
	 \serevise{When $Z^*_{\tiny\mbox{QIP}}$ is not obtained in practice, the uplift payment is calculated using equation $U = Z^*_{\tiny\mbox{QIP}}-{Z_{\text{\tiny C}}}(\bar \gamma)$, i.e., ($\boxtimes$), by replacing $Z^*_{\tiny\mbox{QIP}}$ with its best upper bound $Z'_{\tiny\mbox{QIP}}$ obtained by the MIP solver. The derived uplift payment using $Z'_{\tiny\mbox{QIP}}$ provides an upper bound of the optimal uplift payment, since $Z_{\text{\tiny C}}(\bar\gamma)$ does not change for a given $\bar\gamma$ and $Z'_{\tiny\mbox{QIP}} \geq Z^*_{\tiny\mbox{QIP}}$. The extra amount of uplift payment is equal to $Z'_{\tiny\mbox{QIP}}-Z^*_{\tiny\mbox{QIP}}$. }
\end{remark}

\begin{theorem}\label{thm3}
For $A_1 \in \R ^{m \times n}$, $A_2 \in \R ^{m \times n}$, $E \in \R^{p \times n}$, $b_1\in\R^m$, $b_2\in\R^m$, $d\in \R^p$, and $c \in \R^{n}$, considering the integer optimization problem~\eqref{simeq1}, its linear programming relaxation $z =\min\{c^{\intercal}x| Ax \leq b, Ex \geq p, x\in \R^n\}$ (for notation brevity, we define $\Xe_L=\{x\in \R^n|Ax \leq b\}$), and two tightened linear programming relaxation problems $(P_{{\text{\tiny C}}_1}),(P_{{\text{\tiny C}}_2})$ shown below, 
\begin{eqnarray}
(P_{{\text{\tiny C}}_1}): & \ggy{Z_{\text{\tiny C}_1}} = \min\{c^{\intercal}x|Ex \geq p,x \in \Xe_1\}, \label{pt1}\\
(P_{{\text{\tiny C}}_2}): & \ggy{Z_{\text{\tiny C}_2}} = \min\{c^{\intercal}x|Ex \geq p, x\in \Xe_2\}, \label{pt2}
\end{eqnarray}
in which ${conv(\Xe) \subseteq } \Xe_2 \subseteq \Xe_1 \subseteq \Xe_L$ with $\Xe_1 =\{x\in \R^n|A_1x \leq b_1\}$ and $\Xe_2 =\{x\in \R^n|A_2x \leq b_2\}$. Under these two formulations, we use the \dualvector{dual vector} for constraints $Ex\geq p$ to derive the market clearing price, written as $\bar \gamma_1$ and $\bar \gamma_2$. Then the uplift payment under $\bar \gamma_2$ will be no larger than that under $\bar \gamma_1$ if conditions~\eqref{condition1} and~\eqref{condition2} in Lemma~\ref{lemma1} hold.
\end{theorem}
\begin{proof}
According to Lemma~\ref{lemma1}, the uplift payments under formulations~\eqref{pt1} and~\eqref{pt2} can be calculated as $U_1 =  Z^*_{\tiny\mbox{QIP}}-\ggy{Z_{\text{\tiny C}_1}}(\bar \gamma_1)$ and $U_2 =  Z^*_{\tiny\mbox{QIP}}-\ggy{Z_{\text{\tiny C}_2}}(\bar \gamma_2)$ when conditions~\eqref{condition1} and~\eqref{condition2} hold. Based on strong duality in Lagrangian relaxation, we have $\ggy{Z_{\text{\tiny C}_1}} = \ggy{Z_{\text{\tiny C}_1}}(\bar \gamma_1)$ and $\ggy{Z_{\text{\tiny C}_2}} = \ggy{Z_{\text{\tiny C}_2}}(\textcolor{black}{\bar \gamma_2})$. Since $\Xe_2 \subseteq \Xe_1$, we have $\ggy{\ggy{Z_{\text{\tiny C}_1}} \leq \ggy{Z_{\text{\tiny C}_2}}}$. Thus, we can conclude $U_1 \geq U_2$.
	\end{proof}
	\begin{remark} \label{remark2}
		If $\Xe_2=conv(\Xe)$, then $\bar \gamma_2$ is a convex hull price and $Z^*_{\tiny\mbox{QIP}}-\ggy{Z_{\text{\tiny C}_2}}(\bar \gamma_2)$ is the minimum uplift payment.
	\end{remark}

\subsection{\rere{The detailed algorithm}}
We consider the following ($P$1) as the starting point:
\begin{subeqnarray}
		(P1):Z_{{\text{\tiny C}}} & =\hspace{0.1in} & \hspace{-0.2in}\min_{f,x,u,v,e} \sum_{i\in \G}\sum_{t\in \T}g_{it}  \nonumber \\
		\textrm{s.t.}& & ~\eqref{eqn:sys}, (f_i,x_i,u_i,v_i,e_i) \in {\ggy{\De^1_i}}, \forall i \in \G_1,\nonumber\\
		&& (f_i,x_i,u_i,v_i,e_i) \in {\ggy{\De^2_i}}, \forall i \in \G_2,\nonumber\\
		&& (f_i,x_i,u_i,v_i,e_i) \in {\ggy{\De^3_i}}, \forall i \in \G_3 \cup \G_4.\nonumber
\end{subeqnarray}
In this formulation, we keep the convex hull descriptions for generators in $\G_1 \cup \G_2 \cup \G_3$ in problem ($P$), and relax the convex hull set ${\ggy{\De^4_i}}$ for each generator in $\G_4$ to be ${\ggy{\De^3_i}}$. As stated above, adding EUC formulations~\eqref{model:LP_D3} to all of the generators in $\G_4$ will lead to the perfect formulation, but it will increase the computational complexity significantly. The iterative algorithm is designed to select ``necessary'' generators in $\G_4$ and only add~\eqref{model:LP_D3} of ``necessary'' generators into ($P1$), which can effectively tighten the formulation and improve $Z_{\text{\tiny C}}$ with \yu{a} less computational burden. \ggy{The detailed algorithm is shown in Algorithm~\ref{alg:1}. Here we provide the explanation.}
	
Following the conditions in Lemma~\ref{lemma1}, we first solve ($P1$) and get the \dualvector{dual vector} $\bar \gamma$ of equation~\eqref{eqn:sys}. Then we solve the following ($P2$)\yu{,} which is 
a Lagrangian relaxation problem of ($P1$) based on the given price $\bar \gamma$:
%the {relaxed} profit maximization problem for generators based on the given price $\bar \gamma$:
	\begin{subeqnarray}
	\hspace{-0.1in}(P2):Z_{\text{\tiny C}}(\bar \gamma)& \hspace{-0.15in} =  \bar \gamma^{\intercal}p + \min \limits_{f,x,u,v,e} \sum \limits_{i\in \G}\sum \limits_{t\in T}(g_{it}-\bar \gamma ^{\intercal} E_i x_{it} ) \nonumber \\
	\textrm{s.t.}& (f_i,x_i,u_i,v_i,e_i) \in {\ggy{\De^1_i}}, \forall i \in \G_1,\nonumber\\
	& (f_i,x_i,u_i,v_i,e_i) \in {\ggy{\De^2_i}}, \forall i \in \G_2,\nonumber\\
	& (f_i,x_i,u_i,v_i,e_i) \in {\ggy{\De^3_i}}, \forall i \in \G_3 \cup \G_4.\nonumber
	\end{subeqnarray}	
	
To satisfy conditions~\eqref{condition1} and~\eqref{condition2} in Lemma~\ref{lemma1}, we need to get an integral solution to ($P2$). In ($P2$), since there are no coupling constraints among generators, this problem is equivalent to maximizing profit for each generator independently and summing them together. \threvise{Assume the optimal solution to ($P2$) is $(f^*,x^*,u^*,v^*,e^*)$}. For each generator $\threvise{j} \in \G_1\cup\G_2\cup\G_3$, the values of $u^*_{\threvise{j}}, v^*_{\threvise{j}}, e^*_{\threvise{j}}$ will be binaries for any given price $\bar \gamma$, since ${\ggy{\De^1_{\threvise{j}}}}, {\ggy{\De^2_{\threvise{j}}}}$, and ${\ggy{\De^3_{\threvise{j}}}}$ provide the convex hull descriptions. For each generator ${\threvise{j}} \in \G_4$, we may get fractional values in $u^*_{\threvise{j}}, v^*_{\threvise{j}}, e^*_{\threvise{j}}$ since ${\ggy{\De^3_{\threvise{j}}}}$ is a relaxation of ${\ggy{\De^4_{\threvise{j}}}}$. In this algorithm, we consider the constraints in ${\ggy{\De^4_{\threvise{j}}}}$ as a cutting plane group. When a fractional solution to a particular generator $\threvise{j} \in \G_4$ is obtained in ($P2$), we add this cutting plane group \revise{$\De^4_{\threvise{j}}$} to ($P2$) {\ggy{and ($P1$)}} and replace the objective function $g_{\threvise{j}t}$ with $g'_{\threvise{j}t}$. \revise{In this way, we can} ensure that this generator has integral solutions \revise{when we solve ($P2$)} with any given price $\bar \gamma$. The set \ggy{$\Gamma$ (representing the set of generators in $\G_4$ using the constraints in ${\ggy{\De^4_{\threvise{j}}}}$)} is updated accordingly by adding this generator. Since the \dualvector{dual vector} $\bar \gamma$ could change after updating the formulation for certain generators in $\G_4$, we need to solve ($P1$) again and get an updated $\bar \gamma$. We thus create an inner loop to repeat this process until we could not find an updated price $\bar \gamma$.

Next, we further decrease the uplift payment by tightening ($P1$) utilizing Theorem~\ref{thm3}. The optimal solution to the updated ($P1$) could be infeasible for constraints in~\eqref{model:LP_D3} for some generator $i$ in $\G_4/\ggy{\Gamma}$ since ${\ggy{\De^4_i}}$ is relaxed to ${\ggy{\De^3_i}}$. Given the optimal solution $\threvise{(\hat{f}_i,\hat{x}_i,\hat{u}_i,\hat{v}_i,\hat{e}_i)}$ of each generator $i$ in ($P1$), we can identify infeasible generators by solving the following ($P3$) for each generator $i \in \G_4/\ggy{\Gamma}$ with fractional \threvise{$\hat{u}_i,\hat{v}_i$}, or \threvise{$\hat{e}_i$}: %\threvise{These fractional solutions may exist due to the relaxation of $\De_i^4$.} % solutions regarding to binary variables $u$, $v$, and $w$. 
	\begin{subeqnarray}
 	 &\hspace{-0.6in}(P3): & \hspace{-0.3in} \min_{w_i,z_i,y_i,\theta_i, q_i, \phi_i} 1 \label{sys:total3} \\
	&\hspace{-0.8in}\textrm{s.t.}&\hspace{-0.3in}  (w_i,z_i,y_i,\theta_i, q_i, \phi_i) \in {\ggy{\De^4_i}},\nonumber\\
	 &&\hspace{-0.3in} \threvise{\hat x_{is}} = \hspace{-0.2in}\sum \limits_{tk \in \overline{\T\K}, t\leq s \leq k} q^{s}_{tk},\threvise{\hat f_{is}} = \hspace{-0.2in}\sum \limits_{tk \in \overline{\T\K}, t\leq s \leq k} \hspace{-0.2in} \phi^{s}_{tk},\forall s \in \T,\slabel{eq:map1}\\
	  &&\hspace{-0.3in}\threvise{\hat v_{is}}= \hspace{-0.2in} \sum \limits_{kt \in \overline{\K\T}, t= s}\hspace{-0.2in} z_{kt},  \threvise{\hat e_{is}}=\hspace{-0.2in}\textcolor{black}{\sum \limits_{kt \in \overline{\K\T},k=s}\hspace{-0.2in} z_{kt}+\theta_{s}},\forall s \in \T,\slabel{eq:map2}\\
	  &&\hspace{-0.3in}\threvise{\hat{u}_{is}} = \hspace{-0.2in} \sum \limits_{tk\in \overline{\T\K}^1, k \geq s} \hspace{-0.2in} w_k+ \hspace{-0.15in}\sum \limits_{tk \in \overline{\T\K}^2, t\leq s \leq k} y_{tk},\forall s \in \T, \slabel{eq:map3} 
	\end{subeqnarray}
where constraints~\eqref{eq:map1}-\eqref{eq:map3} build the linear mapping between the optimal solution to EUC formulation \eqref{model:LP_D3} and that of the traditional 3-bin MILP UC formulation \eqref{eq:sg}. \ggy{If the above ($P3$) is infeasible for generator $i$, then this generator $i$ will be included in set $\Gamma$. We update problems ($P1$) and ($P2$) by adding the constraints in ${\ggy{\De^4_i}}$ and replacing $g_{it}$ with $g'_{it}$, which tightens ($P1$) by cutting off the infeasible solution and meanwhile ensure integral solutions in ($P2$) for those generators. Note here that the generators with \yu{an} integral solution to $\hat u_i, \hat v_i$, and $\hat e_i$ are always feasible \yu{in} ($P3$) since the convex hull ${\ggy{\De^4}}$ contains all of the feasible integral solutions.}

Similarly, a new \dualvector{dual vector} $\bar \gamma$ could be obtained after solving the updated ($P1$), and it may lead to extra generators with fractional solutions in updated ($P2$). Thus we create an outer loop to repeat this process iteratively to ensure that we have an optimal solution feasible to the perfect formulation ($P$). \revise{The iterative algorithm terminates when we obtain an solution feasible to ($P$)\yu{,} and an integral solution under the corresponding $\bar \gamma$ is obtained in ($P2$).} The resulting uplift payment under price $\bar \gamma$ is ${U} = {\ggy{Z_{\text{\tiny QIP}}^*}} - \ggy{Z_{\text{\tiny C}}}(\bar \gamma)$.% $Z(\bar \gamma)$. 

\serevise{We can observe that this algorithm satisfies the conditions in Theorem~\ref{thm3} and thus, the uplift payment is non-increasing in each iteration. Meanwhile, the uplift payment is bounded below by the minimum uplift payment. Therefore, the algorithm converges based on the monotone convergence theory.}

\serevise{\begin{remark}
Note here that although the outer loop could guarantee the constraints of our integral formulation~\eqref{model:LP_D3} satisfied, the objective of the final ($P1$) in the outer loop could be still different from that of ($P$), because our integral formulation has a different objective function~\eqref{objeuc}, as compared to that of the original formulation~\eqref{objori1}. Thus, the uplift payment calculation in the algorithm may not reflect the actual uplift payment amount, i.e., $\serevise{U} \not= {{Z_{\text{\tiny QIP}}^*}} - {Z_{\text{\tiny C}}}(\bar \gamma)$ in general under this case. Accordingly, this does not guarantee the monotone decrement of the uplift payment. Therefore, solving ($P2$) is necessary to have Theorem~\ref{thm3} held to ensure convergence.  
\end{remark}
}

	\begin{algorithm} \label{alg:1}
	\KwData{The parameters of the system}
	\KwResult{Convex hull price $\bar \gamma$ and uplift payments ${U}$}
	Initialization: initialize ($P1$) and construct set $\Gamma = \emptyset$ \;
	Solve \ggy{MILP~\eqref{eq:1}} with the optimal objective value ${\ggy{Z_{\text{\tiny QIP}}^*}}$\;
	Relax variables $u,v,e$ to be continues\;
	\Do {\textcolor{black}{n != 0}} 
	{	
		%Solve ($P1$) and get the optimal objective value ${\ggy{Z_{\text{\tiny LP}}^*}}$ and dual price $\bar \gamma$\; 
		\ggy{Solve ($P1$) and get the \threvise{optimal solution ($\hat{f},\hat{x},\hat{u},\hat{v},\hat{e}$)} and \dualvector{dual vector} $\bar \gamma$\;}%, and the corresponding $\ggy{Z_{\text{\tiny C}}}(\bar \gamma)$\; }		
		\label{doloopbegin}
		\Do {\textcolor{black}{m != 0}}
		{
			Set count number $m = 0$\;
			\For{$j \in \G_4/\Gamma$}
			{
				Given \dualvector{dual vector} $\bar \gamma$, solve ($P2$) and get the optimal solution $(f_j^*,x_j^*,u_j^*,v_j^*,e_j^*)$ and the optimal objective value $\ggy{Z_{\text{\tiny C}}}(\bar \gamma)$ \;
				\If{there exists fractional values in terms of \threvise{$u^*_j,v^*_j,e^*_j$ in the optimal solution to ($P2$)} }
				{
					$\Gamma \leftarrow \Gamma \cup \{j\}$, $m \leftarrow m+1$\;
					Replace the objective function $g_{jt}$ in ($P1$) and ($P2$) with  $g'_{jt}$\; \label{replace1}
					Replace the corresponding constraint set in ($P1$) and ($P2$) with $\ggy{\De^4_j}$\; \label{replace2}
				}%{Continue\;}  
				
			}
			\ggy{Solve the updated ($P1$) and get the optimal solution \threvise{($\hat{f},\hat{x},\hat{u},\hat{v},\hat{e}$)} and \dualvector{dual vector} $\bar \gamma$\;}%, and the corresponding $\ggy{Z_{\text{\tiny C}}}(\bar \gamma)$\;}
%$(f^*_j,x^*_j,u^*_j,v^*_j,e^*_j), \forall j \in \G/\Gamma$, $(w_{j'}^*,z_{j'}^*,y_{j'}^*,\theta_{j'}^*, q_{j'}^*, \phi_{j'}^*), \forall j' \in \Gamma$, dual price $\bar \gamma$, and optimal objective value ${\ggy{Z_{\text{\tiny LP}}^*}}$\; \label{replace3}}
		}
		Set count number $n=0$\;
		%solve the updated problem ($P1$) and get the optimal solution $(f^*_i,x^*_i,u^*_i,v^*_i,w^*_i), \forall i \in \G/S$, $(w_i^*,z_i^*,y_i^*,\theta_i^*, q_i^*, \phi_i^*), \forall i' \in S$, dual price $\bar \gamma$ and optimal objective value ${\ggy{Z_{\text{\tiny LP}}^*}}$\;
		\For{$i \in \G_4/\Gamma$}
		{\If{there exists fractional values in terms of \threvise{$\hat{u}_i,\hat{v}_i,\hat{e}_i$ in the optimal solution to ($P1$)}}
			{Solve ($P3$) for generator $i$ and check the feasibility\;
				\If{Infeasible}
				{$\Gamma \leftarrow \Gamma \cup \{i\}$; $n \leftarrow n +1$\;
					Replace the objective function $g_{it}$  in ($P1$) and ($P2$) with  $g'_{it}$\;
					Replace the corresponding constraint set in ($P1$) and ($P2$) with ${\ggy{\De^4_i}}$\;
					%solve the updated problem ($P1$) and get the optimal solution $(f^*_i,x^*_i,u^*_i,v^*_i,w^*_i), \forall i \in \G/S$, $(w_i^*,z_i^*,y_i^*,\theta_i^*, q_i^*, \phi_i^*), \forall i' \in S$, dual price $\bar \gamma$ and optimal objective value ${\ggy{Z_{\text{\tiny LP}}^*}}$\;
				}%{Continue\;}
				}
			%{Continue\;}
		}
		%Update the problem ($P1$)\; 
	}\label{doloopend}
	The {estimated} convex hull price is $\bar \gamma$ and the uplift payment is ${U} = {\ggy{Z_{\text{\tiny QIP}}^*}} - \ggy{Z_{\text{\tiny C}}(\bar \gamma)}$.
	\caption{An iterative algorithm (IA1)}
\end{algorithm} 

\subsection{\rere{Complementary algorithm}}
It can be first observed that the outer loop, which tightens ($P1$) and the inner loop\yu{,} which ensures conditions in Lemma~\ref{lemma1} held in ($P2$)\yu{,} can be exchanged. The exchange may result in a different converging path and lead to a different result. We denote these two as IA1 and IA2. To further tighten the results from IA1/IA2, by taking advantage of the independence of ($P3$) to the other parts of the algorithm, we develop a complementary algorithm, which can be implemented in parallel to check and select the potential candidates from $\G_4/\Gamma$ in the resulting ($P1$) after IA1/IA2 to add the corresponding ${\ggy{\De^4_i}}$, which could further tighten the formulation. 

Assuming there are $M+1$ parallel computing nodes with one master and $M$ slave nodes, the steps of the complementary algorithm, denoted as IAC1 and IAC2, are as follows:
	
	Step 1: Run Algorithm IA1/IA2 on the master node. After it finishes, record the resulting ($P1$), price $\bar \gamma$, set $\ggy{\Gamma}$, and $\ggy{Z_{\text{\tiny C}}(\bar \gamma)}$;
	
	Step 2: Divide the generators in $\G_4/\ggy{\Gamma}$ into $M$ groups: $G'_1, \dots, G'_M$. The resulting ($P1$) and $G'_m$ are distributed to each slave node $m$, $\forall m \in \{1,\dots,M\}$;
	
	Step 3: For each slave node $m \in \{1,\dots, M\}$, each generator $j \in G'_m$ is updated in the resulting ($P1$) in sequence following Steps~\ref{replace1}-\ref{replace2} in IA1/IA2. If the optimal objective value is improved by adding a generator $j$, the resulting ($P1$) in the master node will be updated accordingly and solved to get improved uplift payment $U$ and price $\bar \gamma$;
	
	Step 4: The stopping rules are flexible. The algorithm will stop if one of the following rules is satisfied: i) the master node receives $n$ updated generators from the slave nodes, ii) the time limit $T_{\text{limit}}$ is reached, or iii) all of the slave nodes terminate.
	
The final updated convex hull price $\bar \gamma'$ and $\ggy{Z_{\text{\tiny C}}(\bar \gamma')}$ can be used to calculate the uplift payment as $U' ={\ggy{Z_{\text{\tiny QIP}}^*}} - \ggy{Z_{\text{\tiny C}}(\bar \gamma')}$.

	\section{\rere{Computational Experiments on MISO Instances}} \label{sec:num}
In this section, \ggy{we report the performance of our proposed formulations and algorithms for the MISO system. More specifically, we test the performance of} the proposed tight linear programming formulation, the iterative algorithms, and their complementary parts. \ggy{Each MISO instance includes over }$1,100$ generators with their corresponding bidding information provided. \revise{Among all of the generators, the percentages of $\G_1$, $\G_2$, $\G_3$, and $\G_4$ are around $50\%$, $1\%$, $8\%$, and $41\%$}. \ggy{We randomly select $11$ instances each with \yu{a} $36$-hour operation planning horizon and compute} the corresponding convex hull price and uplift payment. All test instances were run on a dual $8$-processor Intel(R) Xeon(R) CPU E$5$-$2667$ v$4$ @ $3.20$GHz $512$GB with Gurobi $8.0.1$ as the optimization solver. The default MIP optimality gap is set to be 1e-3.
	
	\ggy{We report the} results for the following models and algorithms:
	\begin{itemize}
		\item MIP: the traditional $3$-bin \ggy{MILP} UC model~\eqref{eq:1};
		\item LMP: \ggy{energy price is calculated based on the }dual vector corresponding to the system-wide constraints in the UCED problem when commitment statuses are fixed at their optimal values;
		\item TLP: use the approximated convex hull pricing formulation ($P1$) to obtain the price;
		\item IA1, IA2: approximated convex hull pricing formulation using the algorithms described in Section~\ref{sec:algorithm} to obtain the price;
		\item IAC1, IAC2: approximated convex hull pricing formulation using the iterative algorithm plus its complementary algorithm, \ggy{implemented in a single processor.} The setting for the complementary algorithm is as follows: 1) $M = |\G_4/\Gamma|$, 2) $n=2$, and 3) $T_{\text{limit}}=200s$;
		\item {\serevise{OPT: use the exact convex hull pricing formulation ($P$) to obtain the convex hull price}}.
	\end{itemize}
	
\ggy{We tested the cases with and without transmission constraints, respectively. The results for the cases without transmission constraints are reported in Table~\ref{tab:miso}. In the table, column ``Solution (\$)" represents the optimal cost for the system optimization model, with ``MIP" representing the optimal objective for the MILP UC model~\eqref{eq:1} and others representing the optimal objective for the corresponding linear programming relaxation, column ``Uplift payment (\$)" represents the total uplift payment generated from each approach, column ``Time (s)" represents the computational time in terms of seconds for each approach. Since IAC1 and IAC2 are complementary to IA1 and IA2, we only report additional time required to reach the optimal solution. That is, we use }``$\diamond$" to indicate that IAC1 or IAC2 stops after the stopping rule is satisfied, and ``(+$\delta$)" after ``$\diamond$" to indicate that the complementary algorithm takes $\delta$ extra seconds after the iterative algorithm terminates to reach the optimal solution. ``(+0)" shows that the result of IA1 or IA2 is already optimal. \ggy{Column ``Save (\$)" represents the uplift payment savings as compared to the ``LMP" approach}. The savings of the TLP approach are calculated as the difference of uplift payments between the TLP approach and the LMP approach. Since IA1, IA2, IAC1, and IAC2 tighten the TLP model, the savings of these four methods are represented as the extra savings beyond the TLP approach. For example, in case ``C1", the TLP approach can save \ggy{$\$3,093$} as compared to the LMP approach. The IA1 approach can save extra $\$299$, which means the IA1 approach can save $\$3,093+\$299=\$3,392$ in total as compared to the LMP approach. The ``OPT" provides a \serevise{benchmark} to indicate the efficiency of each approach. \serevise{For instance, our proposed algorithms achieve an exact convex hull price and a minimum uplift payment if the uplift payment derived from our approach is the same as that from the ``OPT" algorithm.} \ggy{Column ``Optimal" indicates if an exact convex hull price  and a minimum uplift payment are achieved by each approach.} \serevise{Finally, column ``Diff (\$/MWh)" represents the average absolute deviation between the exact convex hull and LMP prices. For the cases without transmission constraints, the ``Diff  (\$/MWh)" is calculated as ${\sum_{t\in \T} |\bar\gamma_t-\beta_t|}/{|\T|}$ with the convex hull prices $\bar\gamma_t$ and LMP prices $\beta_t$, $\forall t \in \T$.}

	\begin{comment}
	% Table generated by Excel2LaTeX from sheet 'Sheet2'
	\begin{table}[htbp]
		\centering
		\caption{Add caption}
		\begin{tabular}{ccccccc}
			\toprule
			& MIP Solution (\$) & TPU (\$) & Time(s) & IAU (\$) & Time(s) & Save (\$) \\
			\midrule
			C1    & 41108467  & 1122083  & 24    & 1121741  & 226   & 342  \\
			C2    & 58823929  & 3515838  & 26    & 3514568  & 201   & 1270  \\
			C3    & 71199312  & 1902907  & 18    & 1901669  & 399   & 1237  \\
			C4    & 53498544  & 2737154  & 23    & 2736074  & 541   & 1080  \\
			C5    & 52794818  & 2854509  & 24    & 2849102  & 670   & 5407  \\
			C6    & 61406956  & 2530069  & 48    & 2527029  & 733   & 3039  \\
			C7    & 60539543  & 1345911  & 17    & 1345314  & 116   & 597  \\
			C8    & 59666318  & 2360787  & 19    & 2360239  & 559   & 547  \\
			C9    & 71565907  & 1593596  & 19    & 1589320  & 630   & 4276  \\
			C10   & 51182120  & 3297003  & 20    & 3294583  & 213   & 2419  \\
			\bottomrule
		\end{tabular}%
		\label{tab:addlabel}%
	\end{table}%
	\end{comment}
	% Table generated by Excel2LaTeX from sheet 'Sheet3'
% Table generated by Excel2LaTeX from sheet 'Sheet3'
\begin{table}[htbp]
	\centering
	\caption{\textcolor{black}{Test results for MISO without transmission constraints}}
	\scalebox{0.75}{\begin{tabular}{c|cccccc|c}
			\toprule
			Case  & Model & \tabincell{c}{Solution\\ (\$)} & \tabincell{c}{Uplift Payment \\(\$)} & \tabincell{c}{Time \\(s)} & \tabincell{c}{Save\\ (\$)} & Optimal & \serevise{\tabincell{c}{Diff \\ (\$/MWh) }} \\
\midrule
\multirow{8}[1]{*}{C1} 	&	 MIP   	&	39,986,855	&	-	&	 -     	&	 -     	&	 -  & \multirow{8}[1]{*}{\serevise{0.25}}\\%\multirow{8}[1]{*}{\serevise{9.08}}\\
&	 LMP   	&	-	&	3,521	&	36	&	 -     	&	 N \\
& TLP & 39,980,331 & 428 & 15 & 3,093 & N \\
& IA1 & 39,986,726 & 129 & 241 & +299 & Y \\
& IA2 & 39,986,726 & 129 & 247 & +299 & Y \\
&	 IAC1  	&	39,986,726	&	129	&	 $\diamond$(+0)    	&	+299	&	 Y \\
&	 IAC2  	&	39,986,726	&	129	&	 $\diamond$(+0)     	&	+299	&	 Y \\
&	 OPT   	&	39,986,726	&	129	& \revise{11,214} 	     	&	+299	&	 $\star$ \\
\hline												
\multirow{8}[0]{*}{C2} 	&	 MIP   	&	55,311,277	&	-	&	 -     	&	 -     	&	 - & \multirow{8}[1]{*}{\serevise{0.83}}\\%\multirow{8}[1]{*}{\serevise{29.91}}\\
&	 LMP   	&	-	&	13,242	&	35	&	 -     	&	 N \\
& TLP & 55,291,978 & 3,187 & 15 & 10,055 & N \\
& IA1 & 55,309,361 & 1,916 & 177 & +1,271 & Y \\
& IA2 & 55,309,361 & 1,916 & 194 & +1,271 & Y \\
&	 IAC1  	&	55,309,361	&	1,916	&	 $\diamond$(+0)     	&	+1,271	&	 Y \\
&	 IAC2  	&	55,309,361	&	1,916	&	 $\diamond$(+0)     	&	+1,271	&	 Y \\
&	 OPT   	&	55,309,361	&	1,916	&\revise{10,445} 	      	&	+1,271	&	$\star$  \\
\hline												
\multirow{8}[0]{*}{C3} 	&	 MIP   	&	69,299,295	&	-	&	 -     	&	 -     	&	 - & \multirow{8}[1]{*}{\serevise{0.99}}\\%\multirow{8}[1]{*}{\serevise{35.60}}
&	 LMP   	&	-	&	18,114	&	35	&	 -     	&	 N \\
& TLP & 69,290,109 & 2,890 & 15 & 15,224 & N \\
& IA1 & 69,297,643 & 1,652 & 240 & +1,238 & Y \\
& IA2 & 69,297,517 & 1,778 & 132 & +1,112 & N \\
&	 IAC1  	&	69,297,643	&	1,652	&	 $\diamond$(+0)     	&	+1,238	&	 Y \\
&	 IAC2  	&	69,297,643	&	1,652	&	 $\diamond$(+160)    	&	+1,238	&	 Y \\
&	 OPT   	&	69,297,643	&	1,652	&	\revise{18,975}     	&	+1,238	&	 $\star$ \\
\hline												
\multirow{8}[0]{*}{C4} 	&	 MIP   	&	50,763,422	&	-	&	 -     	&	 -     	&	 - & \multirow{8}[1]{*}{\serevise{0.52}}\\%\multirow{8}[1]{*}{\serevise{18.72}}
&	 LMP   	&	-	&	8,948	&	34	&	 -     	&	 N \\
& TLP & 50,754,612 & 2,033 & 18 & 6,915 & N \\
& IA1 & 50,762,469 & 953 & 431 & +1,080 & Y \\
& IA2 & 50,762,469 & 953 & 390 & +1,080 & Y \\
&	 IAC1  	&	50,762,469	&	953	&	 $\diamond$(+0)     	&	+1,080	&	 Y \\
&	 IAC2  	&	50,762,469	&	953	&	 $\diamond$(+0)     	&	+1,080	&	 Y \\
&	 OPT   	&	50,762,469	&	953	&	\revise{15,185}      	&	+1,080	&	$\star$  \\
\hline												
\multirow{8}[0]{*}{C5} 	&	 MIP   	&	49,946,355	&	-	&	 -     	&	 -     	&	 - & \multirow{8}[1]{*}{\serevise{0.34}}\\ %\multirow{8}[1]{*}{\serevise{12.32}}
&	 LMP   	&	-	&	7,637	&	35	&	 -     	&	 N \\
& TLP & 49,861,587 & 6,046 & 16 & 1,591 & N \\
& IA1 & 49,945,716 & 639 & 493 & +5,407 & Y \\
& IA2 & 49,945,716 & 639 & 769 & +5,407 & Y  \\
&	 IAC1  	&	49,945,716	&	639	&	 $\diamond$(+0)     	&	+5,407	&	 Y \\
&	 IAC2  	&	49,945,716	&	639	&	 $\diamond$(+0)     	&	+5,407	&	 Y \\
&	 OPT   	&	49,945,716	&	639	&\revise{16,260} 	     	&	+5,407	&	$\star$  \\
\hline												
\multirow{8}[0]{*}{C6} 	&	 MIP   	&	58,880,776	&	-	&	 -     	&	 -     	&	 - & \multirow{8}[1]{*}{\serevise{0.76}}\\ %\multirow{8}[1]{*}{\serevise{27.35}}
&	 LMP   	&	-	&	36,630	&	38	&	 -     	&	 N \\
& TLP & 58,861,963 & 3,889 & 15 & 32,741 & N  \\
& IA1 & 58,880,049 & 727 & 475 & +3,162 & Y \\
& IA2 & 58,880,049 & 727 & 597 & +3,162 & Y \\
&	 IAC1  	&	58,880,049	&	727	&	 $\diamond$(+0)   	&	+3,162	&	 Y \\
&	 IAC2  	&	58,880,049	&	727	&	 $\diamond$(+0)    	&	+3,162	&	 Y \\
&	 OPT   	&	58,880,049	&	727	&	\revise{16,098}       	&	+3,162	&	 $\star$ \\
\hline												
\multirow{8}[0]{*}{C7} 	&	 MIP   	&	57,307,363	&	-	&	 -     	&	 -     	&	 - & \multirow{8}[1]{*}{\serevise{1.01}}\\ %\multirow{8}[1]{*}{\serevise{36.31}}
&	 LMP   	&	-	&	51,424	&	37	&	 -     	&	 N \\
& TLP & 57,288,519 & 1,832 & 14 & 49,592 & N \\
& IA1 & 57,306,102 & 1,261 & 584 & +571 & N \\
& IA2 & 57,306,120 & 1,243 & 560 & +589 & Y \\
&	 IAC1  	&	57,306,120	&	1,243	&	 $\diamond$(+38)   	&	+589	&	 Y \\
&	 IAC2  	&	57,306,120	&	1,243	&	 $\diamond$(+0)     	&	+589	&	 Y \\
&	 OPT   	&	57,306,120	&	1,243	&	\revise{11,935}       	&	+589	&	 $\star$  \\
\hline												
\multirow{8}[0]{*}{C8} 	&	 MIP   	&	69,977,540	&	-	&	 -     	&	 -     	&	 - & \multirow{8}[1]{*}{\serevise{0.37}}\\ %\multirow{8}[1]{*}{\serevise{13.91}}
&	 LMP   	&	-	&	11,472	&	35	&	 -     	&	 N \\
& TLP & 69,950,201 & 5,230 & 17 & 6,242 & N \\
& IA1 & 69,977,111 & 429 & 758 & +4,801 & Y \\
& IA2 & 69,977,111 & 429 & 780 & +4,801 & Y \\
&	 IAC1  	&	69,977,111	&	429	&	 $\diamond$(+0)   	&	+4,801	&	 Y \\
&	 IAC2  	&	69,977,111	&	429	&	 $\diamond$(+0)     	&	+4,801	&	 Y \\
&	 OPT   	&	69,977,111	&	429	&	\revise{10,993}       	&	+4,801	&	 $\star$ \\
\hline												
\multirow{8}[1]{*}{C9} &		 MIP   	&	47,889,206	&	-	&	 -     	&	 -     	&	 - & \multirow{8}[1]{*}{\serevise{0.60}}\\ %\multirow{8}[1]{*}{\serevise{21.52}}
&	 LMP   	&	-	&	8,875	&	38	&	 -     	&	 N \\
& TLP & 47,860,497 & 4,042 & 14 & 4,833 & N \\ 
& IA1 & 47,887,537 & 1,669 & 188 & +2,373 & Y \\
& IA2 & 47,887,537 & 1,669 & 234 & +2,373 & Y \\
&	 IAC1  	&	47,887,537	&	1,669	&	 $\diamond$(+0)     	&	+2,373	&	 Y \\
&	 IAC2  	&	47,887,537	&	1,669	&	 $\diamond$(+0)    	&	+2,373	&	 Y \\
&	 OPT   	&	47,887,537	&	1,669	&	\revise{10,363}     	&	+2,373	&	 $\star$ \\
\hline												
\multirow{8}[0]{*}{C10} 	&	 MIP   	&	59,195,531	&	-	&	 -     	&	 -     	&	 - & \multirow{8}[1]{*}{\serevise{0.68}}\\ %\multirow{8}[1]{*}{\serevise{24.65}}
&	 LMP   	&	-	&	11,613	&	36	&	 -     	&	 N \\
& TLP & 59,193,235 & 1,899 & 13 & 9,714 & N \\
& IA1 & 59,194,229 & 1,302 & 108 & +597 & Y \\
& IA2 & 59,194,229 & 1,302 & 115 & +597 & Y \\
&	 IAC1  	&	59,194,229	&	1,302	&	 $\diamond$(+0)     	&	+597	&	 Y \\
&	 IAC2  	&	59,194,229	&	1,302	&	 $\diamond$(+0)     	&	+597	&	 Y \\
&	 OPT   	&	59,194,229	&	1,302	&\revise{9,584} 	     	&	+597	&	 $\star$ \\
\hline												
\multirow{8}[0]{*}{C11} 	&	 MIP   	&	49,628,808	&	-	&	 -     	&	 -     	&	 - & \multirow{8}[1]{*}{\serevise{0.38}}\\ %\multirow{8}[1]{*}{\serevise{13.53}}
&	 LMP   	&	-	&	9,918	&	38	&	 -     	&	 N \\
& TLP & 49,620,385 & 1,448 & 17 & 8,470 & N \\
& IA1 & 49,627,991 & 817 & 372 & +631 & Y \\
& IA2 & 49,627,991 & 817 & 115 & +631 & Y \\
&	 IAC1  	&	49,627,991	&	817	&	 $\diamond$(+0)     	&	+631	&	 Y \\
&	 IAC2  	&	49,627,991	&	817	&	 $\diamond$(+0)     	&	+631	&	 Y \\
&	 OPT   	&	49,627,991	&	817	&\revise{16,269} 	     	&	+631	&	 $\star$ \\			
			\bottomrule
	\end{tabular}}%
	\label{tab:miso}%
\end{table}%	

From Table~\ref{tab:miso}, we can first observe that the TLP approach can effectively reduce the uplift payments with a \yu{shorter} computational time as compared to the LMP approach. The price can be obtained without solving the UC problem. \ggy{Second, we can observe that our proposed alternative algorithms IA1 and IA2 are very effective. They provide extra savings based on the TLP approach. Both IA1 and IA2 can solve nine out of the eleven total cases into optimality. For the remaining two cases, at least one of IA1 and IA2 can achieve {the minimum uplift payment and the exact convex hull price}. In practice, ISOs can run both IA1 and IA2 to improve the chance to obtain the exact convex hull price. More importantly, the algorithm ran very fast and can terminate within $13$ minutes for all cases. Third,} for the cases whose optimal values are not reached by IA1 (or IA2), the complementary algorithm can effectively find out the ``key" generators and add them into IA1 (or IA2) to ensure the optimality. Among the eleven cases at most two generators are added after terminating IA1 (or IA2), which also verified the compactness of IA1 and IA2. What is more, the stopping rules for IAC1 and IAC2 \ggy{can be flexible in terms of adding extra generators into the formulation,} which can provide ISOs the flexibility to get a proper price within a time limit. \serevise{Finally, we can observe that the average absolute deviations between the convex hull prices and LMP prices range from $.25$ to $1.01$.}

\ggy{The results for the cases with transmission constraints (\revise{C10 and C11}) are reported in Table~\ref{table:trans}. \serevise{For the cases with transmission constraints, the price is a vector with dimension $|\T| \times |\mathcal{B}|$, where $|\mathcal{B}|$ is the number of buses in the MISO system, which is over $550$. In column ``Diff (\$/MWh)," we report the average absolute deviation ${\sum_{t\in \T}\sum_{b \in \mathcal{B}} |\bar\gamma_{tb}-\beta_{tb}|}/({|\T|\times|\mathcal{B}|})$, where} $\bar\gamma_{tb}$ and $\beta_{tb}$ are convex hull prices and LMP prices at time $t \in \T$ in bus $b \in \mathcal{B}$. The results are also very promising\yu{,} and similar {observations are obtained}. From Table~\ref{table:trans}, we can observe that both IA1 and IA2 achieve {the minimum uplift payment} and exact convex hull price for the first case. For the second case, IA2 achieves the minimum uplift payment\yu{,} and IA1 achieves {it} with the help of the complementary algorithm. \revise{Compared with the results for Cases C10 and C11 without transmission constraints,} the savings for the cases with transmission constraints are relatively larger, \serevise{and the differences between the exact convex hull prices and LMP prices are also relatively larger}. Meanwhile, it does not take much longer to solve the cases with transmission constraints with all finished in $20$ minutes.}  % we omitted the exact convex hull prices and LMP prices for brevity. To show the difference between the exact convex hull prices and LMP prices for cases C10(T) and C11(T), we calculate the absolute difference between the prices from these two methods at each time period and each location, and report the min absolute difference in column ``Min", the maximum absolute difference in column ``Max", the sum of all of the absolute difference in column ``Total", and the median of all of the absolute difference in column ``Median" in Table~\ref{tab:pricetrans}. From Table~  }

\revise{Finally, from Tables~\ref{tab:miso} and~\ref{table:trans}, by comparing the solving time of our proposed algorithms with the OPT approach, we can observe that our proposed algorithms are much more computationally efficient.}% from more than 20,000 seconds to less than 1,200 seconds with an optimal solution achieved at the end of the algorithm.}

\begin{table}[htbp]
	\centering
	\caption{\textcolor{black}{Test results for MISO with transmission constraints}}\label{table:trans}
	\scalebox{0.75}{\begin{tabular}{c|cccccc|c}
			\toprule
			Case  & Model & \tabincell{c}{Solution\\ (\$)} & \tabincell{c}{Uplift Payment\\ (\$)} & \tabincell{c}{Time\\ (s)} & \tabincell{c}{Save\\ (\$)} & Optimal & \serevise{\tabincell{c}{Diff\\ (\$/MWh)}}\\
			\midrule
			\multirow{8}[1]{*}{C10(T)} & MIP & 61,717,153 & - & 584 & - & - &\multirow{8}[1]{*}{\serevise{3.49}}\\  %\multirow{8}[1]{*}{\serevise{71391}}
			& LMP & - & 1,667,967 & 68 & - & N \\
			& TLP & 61,596,521 & 92,541 & 69 & 1,575,426 & N \\
			& IA1 & 61,602,290 & 87,824 & 1,182 & +4,717 & Y \\
			& IA2 & 61,602,290 & 87,824 & 1,240 & +4,717 & Y \\
			& IAC1 & 61,602,290 & 87,824 & $\diamond$(+0) & +4,717 & Y \\
			& IAC2 & 61,602,290 & 87,824 & $\diamond$(+0) & +4,717 & Y \\
			& OPT & 61,602,290 & 87,824 & \revise{81,630} & +4,717 & $\star$ \\
			\hline												
			\multirow{8}[0]{*}{C11(T)} 	& MIP & 50,071,094 & - & 271 & - & - &\multirow{8}[0]{*}{\serevise{2.19}}\\  %\multirow{8}[0]{*}{\serevise{43703}}
			& LMP & - & 476,190 & 58 & - & N \\
			& TLP & 50,020,529 & 24,538 & 41 & 451,652 & N \\
			& IA1 & 50,030,415 & 23,498 & 512 & +1,041 & N \\
			& IA2 & 50,030,417 & 23,495 & 626 & +1,044 & Y \\
			& IAC1 & 50,030,417 & 23,495 & $\diamond$(+39) & +1,044 & Y \\
			& IAC2 & 50,030,417 & 23,495 & $\diamond$(+0) & +1,044 & Y \\
			& OPT & 50,030,417 & 23,495 & \revise{31,857} & +1,044 & $\star$ \\
			\bottomrule
	\end{tabular}}%
	\label{tab:misotrans}%
\end{table}%

\section{Conclusion} \label{sec:conclusion}
\revise{This paper has provided a customized compact formulation, efficient iterative algorithms, and implementation techniques to solve the convex hull pricing problem for MISO instances. To capture the real-world generator characteristics, we categorize four groups of generators and develop two new convex hull descriptions, in which the EUC formulation provides the convex hull description for the generators with most complicated physical and operational restrictions. More importantly, to tackle the computational complexity issues brought by the large-scale MISO system, iterative algorithms with convergence properties have been proposed. The iterative algorithms innovatively take advantage of the uplift payment calculation process and the cutting plane technology, which can achieve a high-quality solution in a short time. The test results verified the efficiency of our method in both cost savings and solving time reduction. In the next step, we will explore other ways to solve this large-scale convex hull pricing problem, such as decomposition approaches described in~\cite{2014zhao}.}

\section*{Acknowledgements}
\revise{The authors thank the editor and referees for their sincere suggestions on improving the quality of this paper.}

\begin{appendices}
\section{\revise{Proof of Theorem~\ref{thm2}}} \label{app:1}
\begin{proof}
{In this proof, we show that EUC formulation is the dual formulation of a primal dynamic programming formulation (similar to~\cite{guan2018polynomial}). Here, we let $T=|\mathcal{T}|$. First, for the dynamic programming algorithm, we define two value functions $V_\downarrow(t)$ and $V_\uparrow(t)$ for each time period, in which $V_\downarrow (t)$ represents the total net cost (i.e., the total cost minus the revenue) from time $t$ to the end when the generator shuts down at time $t +1$, and $ V_\uparrow (t) $ represents the total net cost from time $t$ to the end when the generator starts up at time $t$. {Basically, we use $V_\downarrow(t)$ and $V_\uparrow(t)$ to keep track of the starting points of the ``OFF" and ``ON" intervals. }}% as shown in Figure \ref{fig:2nd_value}. 
{Accordingly, we let $C(t, k)$ represent the net cost for the ``ON'' period $[t, k]_\Z$, where $k \geq t$. The dynamic programming equations are as follows:}
	\begin{subeqnarray} \label{new_dp1}
		%&& V_\uparrow (t) = \min_{  \substack{ k \in [\min \{t+L-1, \\  T-1\}, T-1]_{\Z} }  } \bigg\{ SD(k-t+1) + C(t,k) +  V_\downarrow (k), C(t,T) +  V_\downarrow (T) \bigg\}, \ \forall t \in [1, T]_{\Z},  \slabel{eqn:Ndp1} \\
		& V_\downarrow (t) = & \hspace{-0.05in}\min_{ k \in [t+\ell_t+1,T]_{\Z} } \hspace{-0.03in} \{ \textcolor{black}{S'}(k-t-1)+  V_\uparrow (k),  0 \}, \nonumber\\
		&& \hspace{1in} \forall t \in [t_0, T-\ell_t -1]_{\Z},  \slabel{eqn:Ndp1} \\
		&V_\downarrow (t) = &0, \ \forall t \in [T-\ell_t, T-1]_{\Z}, \slabel{eqn:Ndp1p}\\
		&V_\uparrow (t) = & \hspace{-0.2in}\min_{  \begin{tiny}\substack{\tiny k \in [t+L_t-1, \\ \min\{t+\overline{L}-1,T-1\}]_{\Z} }\end{tiny}  }\hspace{-0.3in} \{ \textcolor{black}{S}(k \hspace{-0.04in}- \hspace{-0.04in}t\hspace{-0.04in}+\hspace{-0.04in}1) \hspace{-0.05in}+ \hspace{-0.05in} C(t,k)\hspace{-0.04in} +\hspace{-0.04in}  V_\downarrow (k), \hat C(t,T) \},\nonumber\\
		&&  \quad \qquad  \qquad \forall t \in [t_0+\ell_t+1, T-L_t]_{\Z},  \slabel{eqn:Ndp2} \\
		&V_\uparrow (t) = &\hspace{-0.13in}C(t,T), \forall t\hspace{-0.03in} \in \hspace{-0.03in} [\max\{T\hspace{-0.03in}-\hspace{-0.03in}L_t\hspace{-0.03in}+\hspace{-0.03in}1,\hspace{-0.03in} T\hspace{-0.03in}-\hspace{-0.03in}\overline{L}\hspace{-0.03in}+\hspace{-0.03in}1\},\hspace{-0.03in} T]_{\Z}{,} \slabel{eqn:Ndp3}
	\end{subeqnarray}
%where $ C(t,k) $ represents the optimal generation cost (i.e., the objective value of economic dispatch problem) if the {generator} starts up at time $ t $ and shuts down at time $ k + 1$ (i.e., online at $k$).
{where} equations \eqref{eqn:Ndp1} represent that when the {generator} shuts down at time $t+1$, it can either start up again at time $k$ with $k-t-1 \geq \ell_t $ (satisfying min-down time limit) and $k \leq T$ or keep offline until the end. Equations \eqref{eqn:Ndp2} indicate that when the {generator} starts up at time $t$, it can either keep online until time $k$ and shut-down at time $k+1$ with $L_t \leq k-t+1 \leq \overline{L}$ (satisfying min-up and max-up time limit) and $k \leq T-1$ or keep online until the end while satisfying the max-up time limit ($\hat C(t,T)$ represents the net cost if $ t \geq T- \overline{L} +1$, and $+\inf$, otherwise). Equations~\eqref{eqn:Ndp1p} and~\eqref{eqn:Ndp3} describe the backward initial conditions. The objective is
	\begin{eqnarray} \label{new_dp2}
	\Phi \hspace{-0.03in}=\hspace{-0.03in} V_\uparrow\hspace{-0.02in} (\hspace{-0.02in}{0}\hspace{-0.02in})\hspace{-0.03in} := \hspace{-0.45in} \min_{ \tiny t \in [t_0, \min\{\overline{L}\hspace{-0.03in}-\hspace{-0.03in}s_0,T\hspace{-0.03in}-\hspace{-0.03in}1\}]_{\Z}} \hspace{-0.1in}\Big\{\textcolor{black}{S}({t\hspace{-0.03in}+\hspace{-0.03in}s_0}) \hspace{-0.03in}+ \hspace{-0.03in}C(1,t)\hspace{-0.03in}+\hspace{-0.03in} V_\downarrow (t),\hspace{-0.03in} \widetilde{C}(1,T)\Big\},
	\end{eqnarray}
	where $\widetilde{C}(1,T)$ represents the net cost if $\overline{L}-s_0 \geq T$, and $+\inf$, otherwise.

This dynamic program can be written as the following equivalent formulation:
	\begin{subeqnarray} \label{model:LP}
		\hspace{-0.2in}&\max  &\Phi \slabel{eqn:LP_obj}\\
		\hspace{-0.2in}{(\textcolor{black}{w_t})} 	&\hspace{-0.3in}\mbox{s.t.} & \hspace{-0.3in}{\Phi} \leq \textcolor{black}{S}(t{+s_0})  + C(1,t)+ V_\downarrow (t),\nonumber\\
		&& \hspace{0.6in} \forall t \hspace{-0.03in}\in\hspace{-0.03in} [t_0,\min\{\overline{L}\hspace{-0.03in}-\hspace{-0.03in}s_0,T\hspace{-0.03in}-\hspace{-0.03in}1\}]_{\Z}, \slabel{eqn:LP1}\\
		\hspace{-0.2in}(\textcolor{black}{w_t}) & & \hspace{-0.3in} {\Phi} \leq C(1,T), \text{if } \overline{L}-s_0 \geq T,\slabel{eqn:LP2}\\
	\hspace{-0.2in}	{(\textcolor{black}{z_{kt}})}	&&\hspace{-0.3in} V_\downarrow (k) \leq \textcolor{black}{S'}(t-k-1) +  V_\uparrow (t), \nonumber\\
		\hspace{-0.2in}&&  \forall t \in [k\hspace{-0.03in}+\hspace{-0.03in}\ell_t\hspace{-0.03in} +\hspace{-0.03in}1,T]_{\Z}, \forall k \in [t_0, T\hspace{-0.05in}-\hspace{-0.03in}\ell_t \hspace{-0.03in}-\hspace{-0.05in}1]_{\Z}, \slabel{eqn:LP3}\\
	\hspace{-0.2in}	(\theta_t) && \hspace{-0.3in}V_\downarrow (t) \leq 0, \hspace{0.7in} \forall t \in [t_0, T-\ell_t -1]_{\Z},\slabel{eqn:LP3P}\\
	\hspace{-0.2in}	{(\theta_{t})}	&&\hspace{-0.3in} V_\downarrow (t) = 0, \hspace{0.7in} \forall t \in [T-\ell_t, {T-1}]_{\Z},\slabel{eqn:LP3pp}\\
	\hspace{-0.2in}	{(\textcolor{black}{y_{tk}})} 	&  &\hspace{-0.3in} V_\uparrow (t) \leq  \textcolor{black}{S}(k-t+1) +  C(t,k) +  V_\downarrow (k), \nonumber \\
	\hspace{-0.2in}	&& \hspace{0.2in}   \forall k \in [t\hspace{-0.04in}+\hspace{-0.04in}L_t\hspace{-0.04in}-\hspace{-0.04in}1,  \min\{t\hspace{-0.04in}+\hspace{-0.04in}\overline{L}\hspace{-0.04in}-\hspace{-0.04in}1,T\hspace{-0.04in}-\hspace{-0.04in}1\}]_{\Z}, \nonumber\\
	\hspace{-0.2in} && \hspace{0.88in} \forall t \in [t_0\hspace{-0.03in}+\hspace{-0.03in}\ell_t\hspace{-0.03in}+\hspace{-0.03in}1, T\hspace{-0.03in}-\hspace{-0.03in}L_t]_{\Z},  \slabel{eqn:LP2}\\
	\hspace{-0.2in}	{(\textcolor{black}{y}_{tT})} 	&&\hspace{-0.3in} V_\uparrow (t) \leq  C(t,T) ,\nonumber\\
	&& \hspace{0.05in}  \forall t \in [\max\{t_0\hspace{-0.03in}+\hspace{-0.03in}\ell_t\hspace{-0.03in}+\hspace{-0.03in}1,\hspace{-0.03in} T\hspace{-0.06in}-\hspace{-0.06in}\overline{L}\hspace{-0.03in}+\hspace{-0.03in}1\}, T\hspace{-0.03in}-\hspace{-0.03in}L_t]_{\Z}, \slabel{eqn:LP22P}\\	
	\hspace{-0.2in}	{(\textcolor{black}{y}_{tT})}	&& \hspace{-0.3in}V_\uparrow\hspace{-0.02in} (\hspace{-0.01in}t\hspace{-0.01in})\hspace{-0.03in} = \hspace{-0.03in} C\hspace{-0.02in}(t,\hspace{-0.02in}T)\hspace{-0.02in}, \hspace{-0.03in}\forall\hspace{-0.01in} t \hspace{-0.03in}\in\hspace{-0.03in} [\max\{T\hspace{-0.06in}-\hspace{-0.06in}L_t\hspace{-0.03in}+\hspace{-0.03in}1,\hspace{-0.03in} T\hspace{-0.06in}-\hspace{-0.06in}\overline{L}\hspace{-0.03in}+\hspace{-0.03in}1\},\hspace{-0.03in} T]_{\Z}{,} \slabel{eqn:LP4}
	\end{subeqnarray}
where in general	
	%Now the optimal value functions in the dynamic program {framework }become decision variables in the above formulation. {To} obtain the value $V_\uparrow (0) $ under the dynamic programming framework, it is equivalent to maximizing variable $ \Phi $ in the linear program above.
%	Since the above linear program cannot be solved directly as $ C(t,k)$ and $C(t,T)$ are unknown, we will show how to obtain their values in the next step.
%	The formulations of $C(t,k)$ are different for different values of $t$ and $k$. When $ k \leq T-1 $ {and $t>1$, }we have the following formulation to calculate $C(t,k)$ with $(t,k)$ given:
	\begin{subeqnarray} \label{model:ED}
		\hspace{-0.2in}	&& C(t,k) = \min  \sum_{s=t}^{k}  \phi_s  \slabel{eqn:ED_obj} \\
		\hspace{-0.2in} 	\mbox{s.t.}&& \hspace{-0.2in} \pl_s \leq x_s \leq \pb_s, \ \forall s \in [t,k]_{\Z},  x_t \leq \overline{V}^u_t,  x_k \leq \overline{V}^e_k, \slabel{eqn:EDram2}\\
		\hspace{-0.2in}		&&\hspace{-0.2in} x_s - x_{s-1} \leq V_s^u, x_{s-1} \hspace{-0.05in}- x_s \leq V_s^e, \forall s \in [t+1,k]_{\Z}, \slabel{eqn:EDramdo}\\
		\hspace{-0.2in}		&& \hspace{-0.2in}\phi_s \geq a_j x_s + {b_j}, \ \forall s \in [t,k]_{\Z}, j\in  [1,N]_{\Z}, \slabel{eqn:EDlin}
	\end{subeqnarray}
where constraints~\eqref{eqn:EDram2} represent the capacity and start-up/shut-down restrictions, constraints~\eqref{eqn:EDramdo} represent the ramping constraints, and constraints~\eqref{eqn:EDlin} indicate that the objective function is represented by a piecewise linear function with $N$ pieces.  
	
By taking the dual of model~\eqref{model:ED}, we can remove the ``max'' sign and insert the dual formulation into model~\eqref{model:LP} and then obtain an overall dual linear program corresponding to the original dynamic program as follows with $C(t,k)$ as a linking decision variable between~\eqref{model:LP} and the dual of~\eqref{model:ED}: 
	%Now we obtain an {alternative} linear program,  by plugging the dual formulation of {the} economic dispatch problem and {redefining} $C(t,k)$ to be a decision variable in the following model.
	%plug in the dual economic dispatch formulation to our linear program by adding the constraints \eqref{eqn:EDd1} -- \eqref{eqn:EDd5} to \eqref{model:LP} in the following way:
	\begin{subeqnarray} \label{model:LP_ED}
		\hspace{-0.3in}& \max  &  {\Phi} \slabel{eqn:LP_EDobj}\\
		\hspace{-0.3in}& \mbox{s.t.}  &  \hspace{-0.1in} \eqref{eqn:LP1} - \eqref{eqn:LP4}, \slabel{eqn:LP_ED1}\\
		\hspace{-0.3in}{(p_{tk})} \hspace{-0.2in}	& & \hspace{-0.1in} C(1,k) \leq M_1, \hspace{-0.03in} \forall k \hspace{-0.03in}\in\hspace{-0.03in} [t_0\hspace{-0.03in}+\hspace{-0.03in}1, \min\{\overline{L}\hspace{-0.03in}-\hspace{-0.03in}s_0,\hspace{-0.03in}T\hspace{-0.03in}-\hspace{-0.03in}1\}]_{\Z}, \slabel{eqn:LP_ED210}\\
		\hspace{-0.3in}{(p_{tk})} \hspace{-0.2in}	& & \hspace{-0.1in} C(t,k)\hspace{-0.04in} \leq\hspace{-0.04in} M_2,\hspace{-0.04in} \forall k\hspace{-0.04in} \in\hspace{-0.04in} [t\hspace{-0.04in}+\hspace{-0.04in}L_t\hspace{-0.04in}-\hspace{-0.04in}1,  \min\{t\hspace{-0.04in}+\hspace{-0.04in}\overline{L}\hspace{-0.04in}-\hspace{-0.04in}1,T\hspace{-0.04in}-\hspace{-0.04in}1\}]_{\Z}, \nonumber\\
		\hspace{-0.3in}&& \hspace{0.76in} \forall t \in [t_0\hspace{-0.03in}+\hspace{-0.03in}\ell_t\hspace{-0.03in}+\hspace{-0.03in}1, T\hspace{-0.03in}-\hspace{-0.03in}L_t]_{\Z}, \slabel{eqn:LP_ED2}\\
		\hspace{-0.3in}{(p_{tk})} \hspace{-0.2in}	&& \hspace{-0.1in}C(t,T)\hspace{-0.03in} \leq\hspace{-0.03in} M_3, \hspace{-0.03in}\forall t \hspace{-0.03in}\in\hspace{-0.03in} [\max\{t_0\hspace{-0.03in}+\hspace{-0.03in}\ell_t\hspace{-0.03in}+\hspace{-0.03in}1,\hspace{-0.03in} T\hspace{-0.03in}-\hspace{-0.03in}\overline{L}\hspace{-0.03in}+\hspace{-0.03in}1\},\hspace{-0.03in} T]_{\Z}, \slabel{eqn:LP_ED22}\\
		\hspace{-0.3in}{(p_{tk})} \hspace{-0.2in}	& & \hspace{-0.1in} C(1,T) \leq M_4,\text{ if } \overline{L}-s_0 \geq T, \slabel{eqn:LP_ED222}\\
		\hspace{-0.3in}&& \hspace{-0.1in} \text{Constraints in the dual formulation of}~\eqref{model:ED}, \slabel{eqn:LP_ED223}
	\end{subeqnarray}
	 where $M_1, M_2, M_3, M_4$ represent the dual objective function of the economic dispatch problem under four different cases based on the value of $t$ and $k$, with the detailed representation omitted for brevity. 
	 
%The detailed formulations of $M_1,\dots,M_4$ are omitted here for brevity. Meanwhile, constraints~\eqref{eqn:LP_ED223} include all constraints in these four types of dual formulations.
	 
	 %: 1) $t=1$, $k\in [t_0,T-1]_\Z$, 2) $t \in [t_0+\ell_t+1, T-L_t]_{\Z}$,$k \in [\min\{t+L_t-1,T-1\}, \min\{t+\overline{L}-1,T-1\}]_{\Z}$, 3) $ t \in [t_0+\ell_t+1, T]_{\Z}$, $k = T$ and 4) the case in which $t=1$, $k = T$. 

	By taking the dual of the derived linear program~\eqref{model:LP_ED} and cleaning the model, we can obtain model~\eqref{model:LP_D3} (EUC formulation) {as described in Section~\ref{subsub:4}}.

To prove that $\De_i^4$ (i.e., the feasible region for formulation~\eqref{model:LP_D3}) is a convex hull description for single UC, we need to show for any arbitrary linear objective function with it as a feasible region (denoted as formulation $LD$), there exists a corresponding optimal solution binary with respect to decision variables $w$, $y$, and $z$. To show this, based on the optimal solution obtained from the dynamic programming framework, we can construct a corresponding solution $(\textcolor{black}{w}^*, \textcolor{black}{y}^*, \textcolor{black}{z}^*, \theta^*,q^*, \phi^*)$ for $LD$, where ${w}^*$ represents the first shut-down, $y^*$ represents the ``ON" interval, ${z}^*$ represents the ``OFF'' interval, $\theta^*$ represents the final shut-down, and $q^*$ and $\phi^*$ represent the generation amount and cost in the ``ON'' interval, in the optimal schedule provided by the dynamic programming approach. It is easy to check $(\textcolor{black}{w}^*, \textcolor{black}{y}^*, \textcolor{black}{z}^*, \theta^*,q^*, \phi^*)$ is feasible for $LD$ since all physical constraints related to generation amount are derived from the economic dispatch problem~\eqref{model:ED} and the constraints related to binary variables are derived from the dynamic programming framework. Meanwhile, the objective value of the constructed solution to $LD$ is equal to the optimal value of the dynamic programming $V_{\uparrow}(0)$. Then, by the strong duality theorem, $(\textcolor{black}{w}^*, \textcolor{black}{y}^*, \textcolor{black}{z}^*, \theta^*,q^*, \phi^*)$, binary with respect to decision variables $w$, $z$, and $y$, is an optimal solution to $LD$. Thus, the conclusion holds. 
Also, there exists an optimal solution to~\eqref{model:LP_D3} binary with respect to decision variables $w, y,$ and $z$, since the objective function is linear and there exists an optimal solution at the extreme point. 
%Therefore, the constructed solution is binary with respect to decision variables $w$, $z$, and $y$, and optimal for any arbitrary linear objective function. Thus, we proved $\De_i^4$ is the convex hull description for single UC. %Further, since the conclusion holds for any arbitrary linear objective function, it is clear that the optimal solution to EUC formulation~\eqref{model:LP_D3} is binary with respect to decision variables $w$, $z$, and $y$.
\end{proof}
\end{appendices}

	\bibliographystyle{IEEEtran}
	\bibliography{CC_Edge}
	%\bibliography{pricing}
\end{document}